\documentclass{article}
\usepackage[a4paper, left=1in, right=1in, top=1in, bottom=1in]{geometry}
\usepackage{jd}
\usepackage{authblk}
\usepackage{hyperref}

\title{A Separation Between Types of Quantum Oracle Separations}
\date{}

\author[1]{Scott Aaronson}
\author[2]{Adam Bouland\footnote{abouland@stanford.edu}}
\author[2]{Jordan Docter\footnote{jdocter@stanford.edu}}
\author[3]{Barak Nehoran}
\affil[1]{Department of Computer Science, University of Texas Austin}
\affil[2]{Department of Computer Science, Stanford University}
\affil[3]{Department of Computer Science, Columbia University}

\begin{document}
\maketitle

\begin{abstract}

Recent works have demonstrated that quantum oracles have subtle behavior, as access to inverse, conjugate or controlled queries can exponentially change the query complexity of certain tasks. Inspired by these works, we introduce the notion of meta-complexity of quantum relativization. We ask: for any two quantum complexity classes, under which ``types'' of quantum oracles are they equal or separated? Different \textit{pantheons} of oracles (or quantum oracle types, e.g.  unitary vs state, poly- vs superpoly-dimensional, closed under inverse or not) form a partially ordered set based on their power in separating complexity classes. Moreover, two oracle pantheons A and B are separated if there exists a pair of complexity classes that are separated under an oracle from pantheon A but yet the complexity classes are equivalent under all oracles from pantheon B.

We show that this meta-complexity can be nontrivial by giving a complete
classification, within the family of oracle pantheons defined in this paper,
of which models can separate the complexity classes $\PostBQP$ and
$\PreciseBQP$, the exponentially precise variant of $\textsf{BQP}$.  Both
classes equal $\textsf{PP}$ in the unrelativized setting.  Within our taxonomy,
they remain equal relative to real or polynomial-dimensional unitary oracles
and whenever inverse or conjugate access is supplied.  In contrast, we give a
separation relative to forward-only complex diagonal unitaries of
superpolynomial dimension, as well as a separation relative to single-qubit
state-preparation oracles.  In particular, we separate the pantheon of
inverse-less diagonal unitaries from the pantheon of classical oracles.
Our separation is inspired by the recent construction of low-depth multiplicative unitary designs, though for the separation we use a forrelation-like problem with complex variables which may be of independent interest.
We view this as a test case for the meta-complexity of oracles which underscores the subtlety inherent to the relativization of quantum complexity classes.

\end{abstract}

\tableofcontents

\section{Introduction}

Oracles are central to complexity theory, shedding light on the power of different models of computing as well as establishing proof barriers for separating classes \cite{baker1975relativizations}.
Quantum oracles are a natural generalization of classical oracles, where instead of computing a function in a single step, one can instead apply a unitary transformation, or prepare a quantum state, in a black-box manner.
Unitary quantum oracles were first introduced by Aaronson and Kuperberg \cite{aaronsonkuperberg} to study the power of quantum proof systems, and are extremely natural for studying the complexity of quantum tasks like learning properties of quantum transformations \cite{aaronsonbarbados}.
Several recent applications of unitary oracles in complexity theory have included limits on the ability either to amplify quantum proof systems  \cite{aaronson2008perfect,aaronson2025limits} or to reduce $\textsf{QMA}$ to have a unique witness \cite{anshu2026complexity} in a quantum relativizing manner. 
Until recently, it has been an open question whether even the separation that served as the original application of unitary oracles in~\cite{aaronsonkuperberg} could survive relative to a classical oracle, with a long line of works showing that it held relative to various types of oracles, such as under an in-place permutation oracle~\cite{fefferman2015quantum}, a classically-accessible oracle~\cite{li2023classical}, a distributional oracle~\cite{natarajan2024distribution}, and most recently, under a classical oracle~\cite{bostanci2025qcma}.
It is unclear if many of the other results that are known relative to unitary oracles also hold relative to such other oracle types.
Another oracle model that has been of significant recent use in quantum cryptography is that of quantum state oracles~\cite{morimae2024unconditionally,qian2024unconditionally}, which prepare a fixed quantum state in a single timestep. Such oracles have served as an important tool for separating cryptographic primitives in microcrypt~\cite{ananth2024cryptography,chen2025power,bostanci2025oracle,behera2025newworld}, and it has been shown that in certain cases, such results can be upgraded to work under a unitary oracle~\cite{behera2025newworld,goldin2025translating}. 

Recently, there has been an intense interest in the behavior of quantum oracles, exploring how changes to the definitions of quantum oracles -- such as allowing vs disallowing the ability to invert the oracle -- can change the difficulty of quantum tasks for polynomial-time quantum algorithms~\cite{zhandry2025model}. Earlier evidence for this sensitivity was given by \cite{BassirianFeffermanMarwaha}, who showed that the complexity of quantum property-testing problems can depend on whether the underlying function is presented through a standard, in-place, or phase oracle, including randomized variants.
More recently, several works have shown exponential advantages from giving inverse, conjugate \cite{tang2025conjugate}, or controlled \cite{tang2025controlled} unitary access, and shown certain algorithms like amplitude amplification lose their speedups without inverse access \cite{tang2025amplitude}.

\subsection{Our results}

In this work we explore how quantum oracle behavior affects quantum complexity theory.
We initiate the study of the meta-complexity of quantum relativization.
Rather than studying how quantum oracles affect $\textsf{BQP}$, we instead ask how different types of oracles affect the relationships \emph{across} different quantum complexity classes.
For any two quantum complexity classes, we ask: which types of quantum or classical oracles can be used to separate the classes?
The different types of quantum oracles (unitary vs state, inverse-closed or not, real vs. complex, etc.), which we call \emph{oracle pantheons},%
\footnote{
    In Greek, a pantheon is a collection of all the gods common to a religion. Here, we take a pantheon to be a collection of oracles that share some chosen structural properties.
}
form a partially ordered set based on the ease of proving oracle separations.\footnote{In particular, if a pantheon A of quantum oracles contains a pantheon B of quantum oracles, then a separation with respect to pantheon B lifts to pantheon A.}
Thus for any two quantum complexity classes, it is natural to ask if they are equal vs unequal with respect to different pantheons of quantum oracles. We say that two pantheons A and B are separated if there exist a pair of complexity classes such that the complexity classes are separated relative to some oracle from pantheon A, but yet the two complexity classes are equal relative to all oracles of pantheon B. Likewise, two pantheons are equivalent if separations relative to one always imply separations relative to the other, and vice versa.%
\footnote{
    Pantheons that are defined in different ways may nevertheless be equivalent in their ability to distinguish complexity classes. As an example of this, let the first pantheon be the set of in-place permutation oracles which come with access to the (in-place) inverse, and let the second pantheon be the set of out-of-place permutations, also with access to the inverse. Queries to one can be used to simulate queries to the other with only a factor of 2 overhead. Thus, any separation that holds relative to one also holds relative to the other. 
}
This provides a fine-grained view of quantum relativization which does not have a full classical analogue.

To provide an example of this framework, we study the quantum relativization of two closely related quantum complexity classes -- $\PostBQP$ and $\PreciseBQP$.
These classes capture the power of quantum computing with postselection, or with an inverse exponential acceptance probability gap, respectively.
It has been previously shown that both of these classes are equal to $\PP$ with respect to all classical oracles. 
In particular, Aaronson showed that $\PostBQP$ is equal to the classical class $\PP$ \cite{aaronson2005quantum}, and this proof classically relativizes. Likewise Watrous \cite{watrous2008quantum} observed that $\PreciseBQP$ is equal to $\textsf{PP}$ in a classically relativizing manner as well.
Our first observation is that the story becomes more complicated in the quantum oracle setting. In particular, Aaronson's proof that $\PP\subseteq \PostBQP$ does not quantumly relativize. At a high level this is because Aaronson's proof uses the fact that classical oracles allow one to explicitly read out the entries of their corresponding unitary. In contrast with quantum oracles, giving query access to $U$ does not give access to the entries $U_{ij}$, indicating that the story becomes more interesting in quantum setting. 

We show that the meta-complexity of relativization is quite subtle for these
classes by giving a complete classification, within the taxonomy defined
below, of which of these oracle pantheons can separate them.
First, we show that $\PostBQP$ and $\PreciseBQP$ remain identical with respect to unitary oracles of many different types - in particular, if the unitary is either of polynomial dimension, or real, or if one is given complex conjugate or inverse query access. 
On the other hand, we then show that these classes can be separated by a unitary oracle that is complex, superpolynomial dimensional, and without inverse or conjugate access. In fact, we show that even a diagonal unitary of this form can separate them.

\begin{thm}[Informal] 
\label{thm:informal-equality}
$\PostBQP$ and $\PreciseBQP$ are equal with respect to unitary oracles which are real, or of polynomial dimension, or whenever conjugate or inverse queries are allowed.
\end{thm}

\begin{thm}[Informal]
There is a diagonal unitary oracle $U$ such that $\PostBQP^U \subsetneq \PreciseBQP^U$, where $U$ is complex, superpolynomial dimensional, and inverse and conjugate queries are disallowed.
\end{thm}

We also give a related separation in the state-preparation oracle setting, where even a simple complex one-qubit state oracle can separate the classes:

\begin{thm}[Informal] There is a state-preparation oracle that prepares a complex single-qubit state $\ket{\psi}$ such that $\PostBQP^{\ket{\psi}} \subsetneq \PreciseBQP^{\ket{\psi}}$.%
\footnote{
    We use the notation $\textsf{A}^{\ket{\psi}}$ to indicate oracle access to the state-preparation oracle that takes no input and prepares the state $\ket{\psi}$ at each query.
}
\end{thm}
Interestingly, the equality results of \autoref{thm:informal-equality} show this simple single-qubit state separation cannot be lifted to an analogous single-qubit unitary separation.

One interesting aspect of our work is that we give an ``artificial'' quantum oracle separation -- i.e. a quantum oracle separation between classes that cannot be made with respect to a classical oracle (nor in the unrelativized setting). 
Among the closest prior examples of ``artificial'' quantum oracle separations in the literature are the following:

\begin{itemize}
    \item 
    In 2008 Aaronson \cite{aaronson2008perfect} showed a quantum oracle relative to which $\textsf{BQP}$ is not contained in $\textsf{EQEXP}$ (exact quantum exponential time), even though $\textsf{BQP}$ is contained in $\textsf{EQEXP}=\textsf{EXP}$ in a classically relativizing manner.
    \item 
    One can show a quantum oracle separation between $\textsf{QCMA}$ and $\textsf{QCMA}_1$ (i.e. with perfect completeness), even though $\textsf{QCMA}=\textsf{QCMA}_1$ \cite{jordan2011achieving}, by piecing together the results of \cite{jordan2011achieving} and \cite{aaronson2008perfect}.
    \item 
    Agarwal and Kundu established another bounded-error example of this phenomenon: although $\QMA\subseteq \textsf{polyQCPH}$ relative to every classical oracle, they construct a quantum unitary oracle $U$ for which $\QMA^U\nsubseteq \textsf{polyQCPH}^U$ \cite{agarwal2026nonstandardoraclesboundederrorcomplexity}. Their separation persists even when conjugate, transpose, and inverse queries are allowed; in contrast, our results show that the relationship between \(\PostBQP\) and \(\PreciseBQP\) depends sharply on which of these oracle-access modes is available.
\end{itemize}

Our work shows that quantum relativization can depend further on the type of quantum oracle model (state, conjugate, inverse, real, poly vs superpoly dimensional). Moreover, this gives the first formal separation between the pantheon of diagonal unitaries (without complex conjugates) and the pantheons appearing in~\autoref{thm:informal-equality}, as well as a separation between the pantheon of single-qubit state-preparation oracles and those in~\autoref{thm:informal-equality}.

Our results demonstrate that the questions recently raised about modeling quantum oracles \cite{zhandry2025model,tang2025controlled,tang2025amplitude} can have stark implications for quantum complexity theory, as the existence vs nonexistence of quantum relativization barriers can be model-dependent, depending subtly on what type of quantum oracle one is considering.
One cannot say a proof does vs. does not quantumly relativize, without specifying which quantum access model one is referring to.

At a high level, our choice of complexity classes is inspired by the recent construction of low-depth multiplicative unitary t-design construction of Schuster, Haferkamp and Huang \cite{schuster2025random}. 
Intuitively multiplicative designs can fool \PostBQP~algorithms, as they preserve the ratios of probabilities output by quantum algorithms.\footnote{Formalizing this ends up being more complicated than first appears, however!} On the other hand, low-depth designs are easy to distinguish from Haar random unitaries by $\PreciseBQP$ algorithms, because they are relatively low-entanglement, and a quantum algorithm can estimate $\Tr(\rho^2)$ for a subset of the matrix to gain an inverse exponential signal in the amount of entanglement in the construction.
This indicates that $\PostBQP$ might be weaker than $\PreciseBQP$ in unitary oracle settings, which inspired the study of these classes.
While our actual quantum oracle separation does not use \cite{schuster2025random}'s construction -- as it is insufficient for our purposes\footnote{This is because we need to show a bound vs polynomially many queries, and their construction loses its low entanglement properties when $t=\text{poly}(n)$. However, it is possible to prove a much weaker low-bounded query separation using their results, e.g. there are ensembles which $\PostBQP$ needs $\omega(\log n)$ queries to distinguish but $\PreciseBQP$ can distinguish with two queries.} -- our results show the separation is only possible in the precise settings (complex, inverse-free) where the SHH construction works.
We believe this is not a coincidence and might point to a deeper connection between these works.

\subsection{Proof sketch}

For ease of exposition, we first describe our one-qubit state-oracle separation.  Newman's theorem provides useful intuition for why the real and complex settings can behave differently, while our rigorous lower bound applies to the structured ratios of positive polynomials obtained by phase-averaging a postselected algorithm. We then explain why this one-qubit result does not immediately generalize to a one-qubit unitary separation, before describing our no-go results for such generalizations.
Finally we will describe the unitary separation for the complex, inverse-free case, which uses a complex form of Forrelation which might be of independent interest.

\subsubsection{A one-qubit state separation}\label{subsub:on_qubit_sep}
For the one-qubit state separation, it is helpful to recall the connection between $\PostBQP$, $\PreciseBQP$, and low-degree rational and threshold functions. 
Since the work of Beals et al. \cite{beals2001quantum}, quantum algorithms are intricately tied to low degree polynomials  -- i.e. the acceptance probability of a $t$-query quantum algorithm is a degree-$\leq 2t$ polynomial in the oracle entries.
The same fact extends to quantum oracles \cite{she2022unitary}, where the polynomial is now in the amplitudes of the state oracle, or the entries of the unitary oracle.
Adding power to the complexity class corresponding adds power to the corresponding polynomial measure of complexity.
For $\PostBQP$ the corresponding polynomials turn into rational functions.
In fact, for classical oracles, $\PostBQP$ query complexity and approximate rational degree are polynomially related by a result of Mahadev and de Wolf \cite{mahadev2014rationalapproximationsquantumalgorithms}.
For $\PreciseBQP$ the polynomials are upgraded to \emph{threshold polynomials}, ie polynomials which are $<1/2$ on no instances and $\geq 1/2 + 2^{-\text(poly)(n)}$.
Thus if we hope to separate 

In the case of classical oracles, however, no such separation is possible. This is related to a fundamental result in approximation theory known as Newman's theorem, which gives bounded low-degree rational functions that ``jump'' sharply across a small interval of the real line.
This stands in sharp contrast to polynomials which require degree $\Omega(1/\epsilon)$ to jump over such an interval.
Newman's theorem allows rational functions in a real variable to be highly sensitive to small changes in the input value.
This allows one to convert low-degree threshold polynomials -- which are slightly lower/higher on no/yes instances, by an inverse-exponential amount -- into low-degree rational functions by composing with a rational approximant.
As threshold polynomials are related to \PreciseBQP~query complexity, this allows one to show the equivalence of $\PostBQP$ and $\PreciseBQP$ query complexity for classical oracles, up to polynomial factors \cite{mahadev2014rationalapproximationsquantumalgorithms}. 
This gives a way of showing these classes are equivalent in the query model, in addition to being equivalent as complexity classes. 

This suggests a useful heuristic for a quantum-oracle separation.  A real
rational approximant can concentrate its rapid variation near the relevant
transition on the real line.  By contrast, suppose a rational expression in a
complex variable must distinguish two nearby radii uniformly over every
phase.  One might expect rapid variation in every angular direction to be
harder to achieve with only a few poles; \autoref{fig:complex-pole-obstruction}
illustrates this intuition.  We emphasize that we neither need nor prove a
general theorem of this form for arbitrary bounded rational functions on the
disk.  Instead, the state-oracle proof establishes the analogous obstruction
for a structured class of rational functions arising from postselected quantum
algorithms.

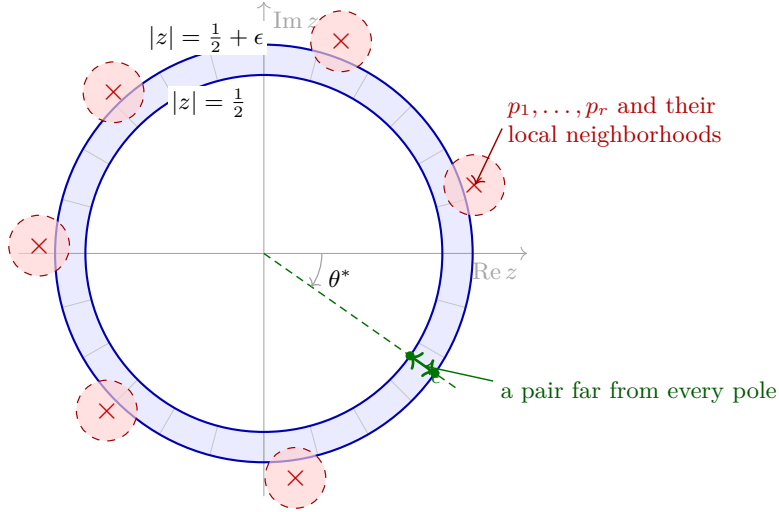
\begin{figure}[H]
\centering
\begin{tikzpicture}[x=1.18cm,y=1.18cm]
  \def\rin{2.00}
  \def\rout{2.34}

  \fill[blue!8,even odd rule]
    (0,0) circle (\rout) (0,0) circle (\rin);
  \foreach \ang in {0,15,...,345}
    \draw[gray!45,line width=0.25pt]
      (\ang:\rin) -- (\ang:\rout);
  \draw[->,gray!70] (-2.75,0) -- (2.95,0)
    node[below left,font=\small] {$\operatorname{Re}z$};
  \draw[->,gray!70] (0,-2.72) -- (0,2.82)
    node[below right,font=\small] {$\operatorname{Im}z$};
  \draw[blue!70!black,thick] (0,0) circle (\rin);
  \draw[blue!70!black,thick] (0,0) circle (\rout);
  \node[fill=white,inner sep=1pt,font=\small]
    at (110:1.8) {$|z|=\frac12$};
  \node[fill=white,inner sep=1pt,font=\small]
    at (105:2.5) {$|z|=\frac12+\epsilon$};

  \foreach \ang/\rad in
    {18/2.48,70/2.53,133/2.47,178/2.52,225/2.49,278/2.54}{
      \begin{scope}[shift={(\ang:\rad)}]
        \fill[red!18,opacity=0.65] (0,0) circle (0.34);
        \draw[red!60!black,dashed,line width=0.45pt]
          (0,0) circle (0.34);
        \node[text=red!75!black,font=\large] at (0,0) {$\times$};
      \end{scope}
    }
  \node[anchor=west,align=left,font=\small,text=red!70!black]
    at (2.63,1.48) {$p_1,\ldots,p_r$ and their\\local neighborhoods};
  \draw[->,red!60!black,line width=0.55pt]
    (2.67,1.40) -- (18:2.49);

  \draw[green!45!black,densely dashed,line width=0.55pt]
    (0,0) -- (325:2.65);
  \draw[<->,green!45!black,line width=1.0pt]
    (325:\rin) -- node[below right=-1pt,font=\small] {$\epsilon$}
    (325:\rout);
  \fill[green!45!black] (325:\rin) circle (1.7pt);
  \fill[green!45!black] (325:\rout) circle (1.7pt);
  \draw[->,gray!75,line width=0.45pt]
    (0.65,0) arc[start angle=0,end angle=-35,radius=0.65];
  \node[font=\small] at (342:0.90) {$\theta^*$};
  \node[anchor=west,align=left,font=\small,text=green!35!black]
    at (2.55,-1.55) {a pair far from every pole};
  \draw[->,green!45!black,line width=0.55pt]
    (2.57,-1.43) -- (325:2.20);
\end{tikzpicture}
\caption{A schematic pole-counting heuristic in the complex plane.  For every
angle $\theta$, the two promise circles contain a radial pair separated by
$\epsilon$.  The picture suggests that finitely many localized regions of
rapid variation (red) may leave a radial pair (green) uncovered. }
\label{fig:complex-pole-obstruction}
\end{figure}

We realize this intuition with a one-qubit state oracle.  Write the state as
\[
    \sqrt{1-\tau}\,e^{i\theta}\ket{0}
    +\sqrt{\tau}\,e^{i\phi}\ket{1}.
\]
For a $t$-query postselected algorithm, average the joint
accept-and-postselect probability and the total postselection probability
separately over $\theta$ and $\phi$.  The pointwise completeness and soundness
inequalities survive this averaging, while the phase average removes cross
terms between distinct monomials.  The two averaged probabilities therefore
have the form
\[
    p(\tau)=\sum_{s,r}A_{s,r}(1-\tau)^s\tau^r,
    \qquad
    q(\tau)=\sum_{s,r}B_{s,r}(1-\tau)^s\tau^r,
\]
where all coefficients are nonnegative and the exponents are $O(t)$.

On an interval bounded away from $0$ and $1$, the logarithmic derivative of
each monomial has magnitude $O(t)$.  Positivity makes the logarithmic
derivative of each sum a convex combination of these monomial derivatives, so
$|(\log p)'|,|(\log q)'|=O(t)$.  Consequently the conditional acceptance ratio
$p/q$ has derivative $O(t)$ in the relevant central interval.  A constant
change across an interval of width $\epsilon$ therefore requires
$t=\Omega(1/\epsilon)$.  This is the rigorous, structured statement behind
the Newman heuristic; it does not assert a lower bound for arbitrary complex
rational functions.

Taking $\epsilon=2^{-n}$ gives a state-oracle problem in $\PreciseBQP$ that
requires exponentially many queries for $\PostBQP$.
See \autoref{sec:state_sep} for the complete proof.

\subsection{No-go for unitary separations: polynomial dimensional}

Now that we have a one-qubit state separation in hand, it is natural to ask if it is possible to directly ``lift'' this to a one-qubit unitary separation. 
The most obvious way is simply to give the unitary preparing the state. 
Concretely, our state oracle separation is achieved with states of the form
\[ \alpha \ket{0} + \beta \ket{1}\]
where $\alpha$ is real, $\beta$  and the hard problem is to determine the norm of $\beta$ is above or below some exponentially close thresholds while its complex phase is uniform.
Why not then simply give access to the unitary preparing this state?
Up to global phase, this unitary is equal to $\left(\begin{matrix} \alpha && -\beta^* \\ \beta && \alpha^*  \end{matrix}\right)$.
This gives an immediate problem: giving access to the unitary preparing the state immediately gives access to the complex conjugates of the state's entries.
Our state-oracle lower bound exploits the restricted polynomial structure that
remains after averaging the unknown phases.  A unitary completion exposes the
conjugate amplitudes as well, allowing mixed expressions such as
$\beta\beta^*=|\beta|^2$.  The positive-polynomial argument above therefore
does not transfer directly.  In terms of the Newman heuristic, conjugate
variables allow the radial information to be represented by a real variable.

At first glance this obstruction appears superficial -- the fact the complex variables appear there ``in plain view'' is a special fact about $SU(2)$.
What if instead we go to 3, 4 or 5 dimensions? 
In these settings, simply writing out the matrix entries does not directly force any particular variable to be the complex conjugate of any other.
Thus, while our 1-qubit unitary approach fails, one might hope that an $O(1)$-qubit separation is possible.

It turns out, however, this is not possible, due to a more fundamental obstruction. 
Namely, surprisingly, we use a complex conjugation algorithm from \cite{miyazaki2019complex} to show that low-dimensional unitary oracles \emph{implicitly contain} information about their conjugate entries with small numbers of queries. This comes from a representation-theoretic identity, but we provide a simple and self-contained proof from first principles, which follows more along the lines of \cite{gavorova2024topological}.
The exact identity underlying the
construction uses $d-1$ queries to $U$ to implement its complex conjugate
$U^*$ up to a global phase.  In the fixed finite-gate model, the surrounding
isometries can be synthesized to error $\delta$ using
$\poly(d,\log(1/\delta))$ gates.  Thus, for any particular entry $U_{ij}$,
$d-1$ queries suffice to obtain $U^*_{ij}$ as an amplitude to error $\delta$;
when $d=\poly(n)$ and $\delta=2^{-\poly(n)}$, the total cost is polynomial in
$n$.
Thus low-dimensional unitaries cannot ``hide'' their conjugate variables from poly-time quantum algorithms.
The reduction from exact implementation of $U^*$ to black-box queries to $U$ takes advantage of the representation-theoretic properties of the unitary group, but we give the algorithm and proof from first principles (without invoking representation theory) in~\autoref{sec:poly_dim_equiv}. The main idea starts from the observation that when $\ket{\phi}$ is the fully anti-symmetric state on $d$ qudits, the corresponding matrix element of the tensor power $\bra{\phi} U^{\otimes d} \ket{\phi}$ is equal to the determinant of $U$. Moreover, we can similarly get the co-factors of $U$ (the determinants of the submatrices where one row and one column are removed) from letting the two states (the bra and the ket in the matrix element) be fully anti-symmetric on all but the relevant row or column. From there, we can see that implementing the complex conjugate of $U$ (up to a global phase) is a variant of the classic Cramer's Rule method for matrix inversion~\cite{cramer1750introduction}.

Using this fact, we then show that $\PostBQP$ and $\PreciseBQP$ are equivalent with respect to polynomial dimensional unitary oracles.
See \autoref{sec:poly_dim_equiv} for the detailed proof.

\subsection{No-go for unitary separations: real, conjugate, or inverse access}

The equivalence with conjugate or inverse access is consistent with the Newman
heuristic: these models make conjugate entries, and hence real combinations of
complex variables, accessible.  Our formal proof does not invoke a complex
analogue of Newman's theorem.  Instead, it uses coherent reversal for inverse
access and a postselected gate-teleportation simulation for conjugate access.
Containment of \PostBQP in \PreciseBQP relative to unitary oracles with
conjugate and inverse access is immediate.

For the other direction, assuming inverse access, the key observation is that the acceptance probability, $a$, of the \PreciseBQP~circuit can be encoded as an amplitude by running the circuit in reverse. Postselection then allows us to determine whether $a$ is above the completeness threshold or lies below the soundness threshold.

For conjugate access, it is enough to obtain inverse access on a postselected
branch. A transposition map \cite{quintino2019probabilistic} using a gate-teleportation procedure takes one query to $\bar U$  and applies $U^\dagger$, conditioned on an outcome of
inverse-exponential probability.

\subsection{A unitary separation with complex, superpoly dimensional oracles}

We now describe our unitary separation between $\PostBQP$ and $\PreciseBQP$. By our prior results we must leverage complex numbers, the lack of conjugate variables or inverse access, and the large dimensionality of the unitaries.
We will achieve this by introducing a complex generalization of Forrelation which is essentially a real inner product between two functions over $C^d$ for a large value of $d$. 

In particular, we construct two distributions over diagonal unitary oracles which \PreciseBQP~can distinguish when viewed as a promise problem, but \PostBQP~cannot on average. This reduces to a unitary oracle separation using standard diagonalization arguments. Let $(w_1,\dots,w_d, v_1,\dots, v_d ) \in \T^{2d}$ be the diagonal entries of a $2d$-dimensional unitary. Let $D_a$ denote the normalized Haar measure over diagonal unitaries conditioned on the level-set equation
\[
  \operatorname{Re}\langle w,v\rangle = \operatorname{Re} \sum_{i=1}^d \overline{w_i}v_i =a.
\]

We consider a Forrelation-like problem over complex variables: the problem of distinguishing between the distributions $D_a$ and $D_{-a}$ for $a = O(\log d)$. Note that this relation highly depends on the imaginary component of complex variables, which we know to be a necessary ingredient for a separation. A \PreciseBQP~algorithm can trivially distinguish these distributions as a promise problem, since there is a quantum algorithm (see \autoref{fig:unitary-preciseBQP-circuit}) with acceptance probability $\dfrac12+\dfrac{a}{2d}$ for $ U\in\operatorname{supp}(D_a)$ and $\dfrac12-\dfrac{a}{2d}$ for $ U\in\operatorname{supp}(D_{-a})$.
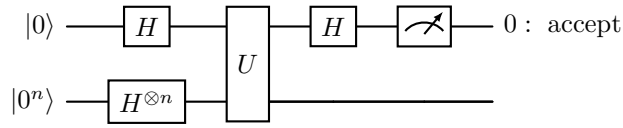
\begin{figure}[H]
\centering
\begin{quantikz}[row sep=0.45cm, column sep=0.55cm]
  \lstick{$\ket{0}$}
    & \gate{H}
    & \gate[2]{U}
    & \gate{H}
    & \meter{}
    & \rstick{$0:\ \mathrm{accept}$} \\
  \lstick{$\ket{0^n}$}
    & \gate{H^{\otimes n}}
    &
    & \qw
    & \qw
    & \qw
\end{quantikz}
\caption{The one-query \PreciseBQP~test.}
\label{fig:unitary-preciseBQP-circuit}
\end{figure}
The main work is to show that no \PostBQP~algorithm can distinguish these distributions as a promise problem. We give a lower bound on the query complexity via the polynomial method. The acceptance probability of a $t$-query quantum algorithm is represented by a bounded degree-$2t$ sum-of-squares polynomial. With the power of postselection, the acceptance probability is represented by a ratio of degree-$2t$ polynomials $\frac{p}{p+q}$. Although the polynomials $p$ and $q$ may depend on all $2d$ entries $(w_1,\dots,w_d, v_1,\dots, v_d )$ of the diagonal, each term/monomial only depends on $2t$ coordinates. The proof begins with a symmetrization step which reduces the polynomial to a more structured polynomial in fewer variables. Then, using the structure, we prove a relative error bound which is strong enough to bound the derivative of the relevant rational function. The lower bound then follows from standard polynomial method arguments generalized to rational functions.
\paragraph{Symmetrization.} 
By the same observations in \autoref{subsub:on_qubit_sep}, it suffices to bound the ratio of the expectations of the polynomials  $\frac{\BE p}{\BE p+\BE q}$(rather than the expectation of ratios). We can therefore use straightforward symmetrization techniques to simplify the polynomials in the numerator and denominator separately. 

The distributions are invariant to the action of multiplying each pair $v_i,w_i$ by the same phase. 
By this symmetry, we show that $\BE p$ (and $\BE q$) only depend on the variables $z_i \coloneqq \overline{w_i}v_i$. 
Additionally, averaging over the independent $w_i$-variables eliminates the cross terms in each square. The result of the symmetrization and averaging is that 

$\BE p$ (and $\BE q$) reduces to a sum of squares of $t$-variate, degree-$2t$ polynomials $p_j$ in the variables $z_1, \dots , z_d$. It is integral to our proof that each sub-polynomial only depends on at most $t$ of the variables $z_i$, and that its square $|p_j|^2$ is \textit{non-negative}. After changing variables accordingly, the condition defining the distributions becomes
\[
\sum_{i=1}^d \operatorname{Re}(z_i)= \pm a.
\]
\paragraph{Relative error bound.}
Since each term is non-negative, it suffices to upper bound the relative error of each term individually. This relative error can be applied to every nonnegative term in the acceptance probability which in turn yields a relative error for $\BE p$ and $\BE q$ under \(D_a\) and \(D_{-a}\). This multiplicative error is necessary to bound error with respect to the \PostBQP~accepting rational function $\frac{p}{p+q}$. 

We show multiplicative closeness by proving that the underlying conditional distribution are also multiplicatively close. Again this suffices by the fact the polynomials are nonnegative.
More generally, the expectations of bounded nonnegative functions must be multiplicatively close if their underlying distributions are multiplicatively close, and our polynomials are bounded because they are acceptance probabilities of quantum algorithms, so this suffices.

To show the distributions are close, we compare the conditional measures on a function \(H\geq 0\) that depends on only \(k\leq t\) coordinates. Since the contribution to the real value $\pm a$ is distributed across all $d$ terms $\overline{w_i}v_i$, each $H$ cannot be ``too sensitive'' to changing distributions $D_a$ and $D_{-a}$ on average. 

For each nonnegative term, we construct the marginal distributions over the $d-t$ (or more) remaining $z_i$ for each of the distributions and show that they are multiplicatively close. Let $f_m$ denote the density of
\[
  S_m=\sum_{i=1}^m\cos\theta_i, \qquad \theta_1,\ldots \theta_m ~\text{i.i.d. on} ~[0,2\pi)
\]
We  consider the marginal density conditioned on the value $s=\sum_{i=1}^k\operatorname{Re}z_i$ taken by the $k$ coordinates which $H$ depends on:
\begin{align*}
  f_{d-t}(a-s) \\
  f_{d-t}(-a-s) = f_{d-t}(a+s) \tag{$f_m$ is even }.
\end{align*}
We show that their ratio $\frac{f_{d-t}(a-s)}{f_{d-t}(a+s)}$ is bounded as $ C_{d,t,a} \leq O(\exp (\frac{4|a|t}{d-t}))$.
The intuitive reason this ratio is bounded is that $f_{d-t}$ is extremely flat near the origin when $d \gg t$. The central limit theorem says that this random variable is well approximated by a very broad Gaussian with variance $(d-t)/2$, so its width is of order $\sqrt{d-t}$. Thus, in the central region, moving from $a-s$ to $a+s$ amounts to a very small shift in density. See \autoref{fig:relative-error-intuition}. This relative error can be applied to every nonnegative term in the acceptance probability which in turn yields a relative error for the average values of $p$ and $q$ under \(D_a\) and \(D_{-a}\). The multiplicative error is necessary to bound error with respect to the \PostBQP~accepting rational function $\frac{p}{p+q}$. 

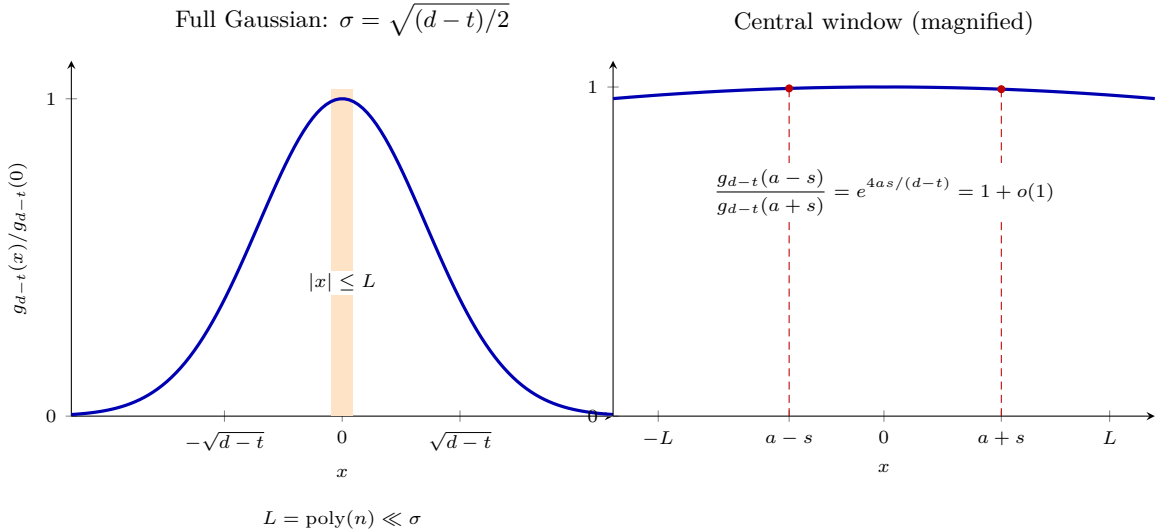
\begin{figure}[H]
\centering
\begin{tikzpicture}
\begin{axis}[
  name=fullgauss,
  width=0.45\textwidth,
  height=4.7cm,
  scale only axis,
  axis lines=left,
  xmin=-2.3, xmax=2.3,
  ymin=0, ymax=1.12,
  xtick={-1,0,1},
  xticklabels={$-\sqrt{d-t}$,$0$,$\sqrt{d-t}$},
  ytick={0,1},
  yticklabels={$0$,$1$},
  xlabel={$x$},
  ylabel={$g_{d-t}(x)/g_{d-t}(0)$},
  tick label style={font=\scriptsize},
  label style={font=\scriptsize},
  title={Full Gaussian: $\sigma=\sqrt{(d-t)/2}$},
  title style={font=\small},
  clip=false
]
  \path[fill=orange!22] (axis cs:-0.09,0) rectangle (axis cs:0.09,1.03);
  \addplot[blue!70!black,very thick,domain=-2.3:2.3,samples=120]
    {exp(-x^2)};
  \node[font=\scriptsize,align=center,fill=white,inner sep=1pt]
    at (axis cs:0,0.42) {$|x|\leq L$};
  \node[font=\scriptsize,align=center]
    at (axis description cs:0.5,-0.29)
    {$L=\operatorname{poly}(n)\ll\sigma$};
\end{axis}

\begin{axis}[
  at={([xshift=0.8cm]fullgauss.east)},
  anchor=west,
  width=0.45\textwidth,
  height=4.7cm,
  scale only axis,
  axis lines=left,
  xmin=-1.2, xmax=1.2,
  ymin=0, ymax=1.08,
  xtick={-1,-0.42,0,0.52,1},
  xticklabels={$-L$,$a-s$,$0$,$a+s$,$L$},
  ytick={0,1},
  yticklabels={$0$,$1$},
  xlabel={$x$},
  tick label style={font=\scriptsize},
  label style={font=\scriptsize},
  title={Central window (magnified)},
  title style={font=\small},
  clip=false
]
  \addplot[blue!70!black,very thick,domain=-1.2:1.2,samples=100]
    {exp(-0.025*x^2)};
  \draw[densely dashed,red!75!black]
    (axis cs:-0.42,0) -- (axis cs:-0.42,{exp(-0.025*(-0.42)^2)});
  \draw[densely dashed,red!75!black]
    (axis cs:0.52,0) -- (axis cs:0.52,{exp(-0.025*(0.52)^2)});
  \fill[red!75!black]
    (axis cs:-0.42,{exp(-0.025*(-0.42)^2)}) circle[radius=1.5pt];
  \fill[red!75!black]
    (axis cs:0.52,{exp(-0.025*(0.52)^2)}) circle[radius=1.5pt];
  \node[font=\scriptsize,align=center,fill=white,inner sep=2pt]
    at (axis cs:0,0.68)
    {$\displaystyle
      \frac{g_{d-t}(a-s)}{g_{d-t}(a+s)}
      =e^{4as/(d-t)}=1+o(1)$};
\end{axis}
\end{tikzpicture}
\caption{If $a$ and $t$ are polynomial in the input length $n$, while $d$ is exponential in $n$, then
$a-s$ and $a+s$ lie in a window of width $L=\operatorname{poly}(n)$, which
is much smaller than $\sigma$.  The density is therefore almost constant between the
two points.}
\label{fig:relative-error-intuition}
\end{figure}

\paragraph{Query lower bound.} Correctness gives $p\geq2q$ on
oracles from $D_a$ and $p\leq q/2$ on oracles from $D_{-a}$.  Averaging gives
\[
  \BE_{D_a}p \geq 2 \BE_{D_a}q,
  \qquad
  \BE_{D_{-a}}p\leq\frac{1}{2}\BE_{D_{-a}}q.
\]
Thus the accepting-to-rejecting ratio must change by a factor of at least four.
This ratio changes by at most $C_{d,t,a}^2$, so correctness requires $O(\exp (\frac{4|a|t}{d-t})) \geq 2$, which gives a lower bound on the number of queries $t$ needed by any
\PostBQP~algorithm to distinguish the distributions.

\subsection{Discussion}

In summary, our work has shown that different models of quantum oracles can change the relativizing properties of quantum complexity classes. A natural future direction is to extend this classification to other quantum complexity classes.
For example, as mentioned earlier, there is a simple one-qubit unitary oracle separation between $\textsf{QCMA}$ and $\textsf{QCMA}_1$ (i.e. with perfect completeness), even though $\textsf{QCMA}=\textsf{QCMA}_1$ \cite{jordan2011achieving}, by piecing together the results of \cite{jordan2011achieving} and \cite{aaronson2008perfect}.
This completely resolves the relativization of $\QCMA$ vs $\QCMA_1$ under our choice of unitary oracle pantheons, though the state question remains open.
Moreover, we pose a challenge to separate each pair of pantheons among the list in \autoref{thm:informal-pantheon-sep}. Each such pantheon separation would indicate that relativization results under one pantheon will not transfer to the other. Any equivalence among these pantheons would be surprising, and indicate that one can in principle work with either pantheon.
We hope that our work serves as inspiration for further research into the subtleties of quantum relativization, and the limitations it places on complexity theory.

\section{Preliminaries}

Throughout, $n$ denotes the input length, $t$ denotes the number of oracle
queries, and $d$ denotes the dimension parameter specified in the statement at
hand.  In particular, when an oracle acts on an additional block qubit and a
$d$-dimensional index register, the full oracle dimension is $2d$. All circuit
families are polynomial-time uniform and use a fixed finite universal gate set
with algebraic entries, which we take without loss of generality to be closed
under adjoints.

\begin{defn}[\PostBQP]
    A promise problem $A=(A_{\mathrm{yes}},A_{\mathrm{no}})$ is in
    \PostBQP~if there exist a polynomial $p$ and a polynomial-time uniform
    family of polynomial-size quantum circuits $\{C_n\}_{n\geq 1}$, each with
    a designated postselection qubit $P$ and output qubit $O$, such that for
    every $x\in (A_{\mathrm{yes}}\cup A_{\mathrm{no}})\cap\{0,1\}^n$, when
    $C_n$ is applied to $\ket{x}\ket{0^{q(n)}}$ for some polynomial $q$:\footnote{This is the inverse-exponential formulation used, for example, in
    \cite{hiromasa_et_al:LIPIcs.TQC.2023.9}.  In the standard finite-gate
    circuit model it is equivalent to Aaronson's original requirement that the
    postselection event have nonzero
    probability~\cite{Aaronson2005QuantumCP}; stating the bound explicitly is
    important here because that equivalence need not hold relative to an
    arbitrary unitary oracle.}
    \begin{enumerate}
        \item[(i)] $\Pr[P=1]\geq 2^{-p(n)}$;
        \item[(ii)] if $x\in A_{\mathrm{yes}}$, then
        $\Pr[O=1\mid P=1]\geq 2/3$;
        \item[(iii)] if $x\in A_{\mathrm{no}}$, then
        $\Pr[O=1\mid P=1]\leq 1/3$.
    \end{enumerate}
\end{defn}
\begin{defn}[\PreciseBQP($c, s$)]
    Let $c,s\colon\mathbb{N}\to[0,1]$ be polynomial-time computable.  A
    promise problem $A=(A_{\mathrm{yes}},A_{\mathrm{no}})$ is in
    $\PreciseBQP(c,s)$ if there is a polynomial-time uniform family of
    polynomial-size quantum circuits $\{V_n\}_{n\geq 1}$, each with a
    designated output qubit, such that on input
    $\ket{x}\ket{0^{q(n)}}$, for some polynomial $q$,
    \begin{description}
        \item[$\bullet$ Completeness:] If $x\in A_{\mathrm{yes}}\cap
        \{0,1\}^n$, then $V_n$ accepts with probability at least $c(n)$.
        \item[$\bullet$ Soundness:] If $x\in A_{\mathrm{no}}\cap
        \{0,1\}^n$, then $V_n$ accepts with probability at most $s(n)$.
    \end{description}
    We define for polytime computable $c,s$, and polynomial $p$,
    \[
        \PreciseBQP
        :=\bigcup_{\substack{c,s\colon
                    c(n)-s(n)\geq 2^{-p(n)}}}
          \PreciseBQP(c,s).
    \]
    This is the exponentially small-gap formulation of
    \PreciseBQP~\cite{gharibian_et_al:LIPIcs.MFCS.2018.58}.
\end{defn}

\begin{lem}[Finite-gate approximation under postselection]
    \label{lem:finite_gate_postselection}
    Let $C$ and $\widetilde C$ be circuits whose output states $\rho$ and
    $\widetilde\rho$ satisfy
    \[
        \frac12\lVert \rho-\widetilde\rho\rVert_1\leq\delta.
    \]
    Let $P$ be a postselection event with $\Pr_C[P]\geq\mu>\delta$, and let
    $A$ be an acceptance event.  Then
    \[
        \Pr_{\widetilde C}[P]\geq\mu-\delta
        \qquad\text{and}\qquad
        \left|
        \Pr_C[A\mid P]-\Pr_{\widetilde C}[A\mid P]
        \right|
        \leq \frac{2\delta}{\mu}.
    \]
    Moreover, suppose that an ideal polynomial-size circuit uses uniformly
    specified constant-qubit gates whose entries are polynomial-time
    computable.  It can be compiled over the fixed gate set to output trace
    distance at most $\delta$ using
    $\poly(n,\log(1/\delta))$ gates.  In particular, inverse-exponential
    accuracy has only polynomial overhead.
\end{lem}

\begin{proof}
    Write $p=\Pr_C[P]$, $q=\Pr_C[A\cap P]$, and use tildes for the
    corresponding probabilities under $\widetilde C$.  Trace distance bounds
    the change in the probability of every measurement outcome, so
    $|p-\widetilde p|\leq\delta$ and
    $|q-\widetilde q|\leq\delta$.  Since
    $0\leq\widetilde q/\widetilde p\leq1$,
    \begin{align*}
        \left|\frac{q}{p}-\frac{\widetilde q}{\widetilde p}\right|
        &\leq \frac{|q-\widetilde q|}{p}
        +\frac{\widetilde q}{\widetilde p}
          \frac{|p-\widetilde p|}{p}
        \leq \frac{2\delta}{\mu}.
    \end{align*}
    For the final statement, synthesize each of the polynomially many ideal
    gates to accuracy $\delta/L$, where $L$ is the circuit size.  The standard
    hybrid argument bounds the total error by $O(\delta)$, while the
    Solovay--Kitaev theorem gives overhead polynomial in
    $\log(L/\delta)$ for each gate~\cite{dawson2006solovaykitaev}.
    Rescaling $\delta$ by a constant absorbs the $O(\delta)$ factor.
\end{proof}

\begin{defn}[Unitary oracle access models]
    A unitary oracle is an infinite sequence $U=\{U_n\}_{n\geq 1}$, where
    $U_n$ acts on $p(n)$ qubits for a fixed polynomial $p$.  On inputs of
    length $n$, only $U_n$ may be queried.

    Formally, assume a quantum algorithm state has the form 
    \[
        |\Phi\rangle = \sum_{z} \alpha_z |z\rangle |\varphi_z\rangle,
    \]
    where $|z\rangle$ is a workspace register and $|\varphi_z\rangle$ is a
    $p(n)$-qubit query register.  A query applies $I\otimes U_n$, mapping this
    state to
    \[
        |\Phi'\rangle
        = \sum_{z} \alpha_z |z\rangle U_n |\varphi_z\rangle.
    \]
    This is the only access supplied by default.  The other access models grant
    additional query gates.  Controlled access grants
    \[
        \operatorname{ctrl}(U_n)
        = (\ket{0}\!\bra{0})_C\otimes I
          +(\ket{1}\!\bra{1})_C\otimes U_n,
    \]
    where $C$ is a control qubit.  Inverse access grants $U_n^\dagger$, and
    conjugate access grants $\bar U_n$, the entrywise complex conjugate of
    $U_n$ in the computational basis.  None of these additional gates is
    available unless explicitly stated.  This is the standard unitary-oracle
    model of \cite{aaronsonkuperberg}; the distinctions among ordinary,
    controlled, and inverse access are discussed systematically in
    \cite{zhandry2025model}, and conjugate access is studied in
    \cite{tang2025conjugate}.
\end{defn}
We will often refer to $U_n$ simply as $U$ when it is clear from context.
Let $\mathcal C$ be a quantum complexity class.  We write $\mathcal C^U$ for
the class of promise problems solvable by a $\mathcal C$ machine with unit-cost
query access to $U_n$ on inputs of length $n$.  The number of queries is
subject to the resource bound defining $\mathcal C$ and hence is polynomial
for the classes considered here.  Similarly, $\mathcal C^{U,U^\dagger}$
(respectively, $\mathcal C^{U,\bar U}$) additionally grants unit-cost query
access to $U_n^\dagger$ (respectively, $\bar U_n$).

\begin{defn}[Quantum state oracle access model]
    Let $\Psi=\{\ket{\psi_n}\}_{n\geq 1}$ be a sequence in which
    $\ket{\psi_n}$ is a $p(n)$-qubit state for some fixed polynomial $p$.  The
    corresponding state oracle
    $\CO_\Psi=\{\CO_{\psi_n}\}_{n\geq 1}$ supplies a fresh copy of
    $\ket{\psi_n}$ on each query.  Formally, for every state $\rho_A$ already
    held by the algorithm,
    \[
        \CO_{\psi_n}(\rho_A)
        =\rho_A\otimes\ket{\psi_n}\!\bra{\psi_n}_R,
    \]
    where $R$ is a fresh $p(n)$-qubit output register.  Equivalently, this is
    the state-generation isometry
    $\ket{\phi}_A\mapsto\ket{\phi}_A\ket{\psi_n}_R$.  The oracle provides no
    adjoint or other purification access.  This is the common reference
    quantum state oracle model; see, for example,
    \cite{bostanci2025oracle}.
\end{defn}

We also introduce notation to define the oracle pantheons as follows:
\begin{itemize}
    \item $\mathsf{D} = $ diagonal unitaries without complex conjugate or inverse access (transpose comes for free for diagonal unitaries)
    \item $\mathsf{D}^* = $ diagonal unitaries with complex conjugate / inverse access
    \item $\mathsf{U} = $ unitaries without complex conjugate, transpose, or inverse access
    \item $\mathsf{U}^* = $ unitaries with complex conjugate access
    \item $\mathsf{U}^T = $ unitaries with transpose access
    \item $\mathsf{U}^\dagger = $ unitaries with inverse access
    \item $\boxed{\mathsf{U}} = $ unitaries with complex conjugate, transpose, and inverse access
    \item $\mathsf{U}_{\mathsf{poly}} = $ unitaries of polynomial dimension without complex conjugate, transpose, or inverse access
    \item $\mathsf{R} = $ orthogonal (real-valued) unitaries without transpose / inverse access (complex conjugate access comes for free)
    \item $\mathsf{F} = $ classical functions, or $\pm 1$ diagonal unitaries (complex conjugate, transpose, and inverse access all come for free)
    \item $\mathsf{S} = $ quantum state preparation oracles (equivalently, isometries with trivial input dimension).
    \item $\mathsf{S}_{2} = $ single-qubit (2-dimensional) quantum state preparation oracles.
\end{itemize}

In this notation, we show the following results:
\begin{thm}[Informal]
    \label{thm:informal-pantheon-sep}
    There is a separation between the oracle pantheons $A$ and $B$ for every $A \in \{\mathsf{D}, \mathsf{U}, \mathsf{U}^T, \mathsf{S}, \mathsf{S}_2\}$ and every $B \in \{\mathsf{U}_{\mathsf{poly}}, \mathsf{R}, \mathsf{D}^*, \mathsf{U}^*, \mathsf{U}^\dagger, \boxed{\mathsf{U}}, \mathsf{F}\}$. 

\end{thm}

\section{Equivalence with inverse or conjugate access}

\subsection{Inverse access}
Inverse access lets us run a circuit backwards.  This allows us to write its
acceptance probability into an amplitude, which postselection can compare with
an exponentially close threshold.

\begin{restatable}{thm}{thmEqualWithInverse}\label{thm:equal_with_inverse}
    For any quantum unitary oracle $U$, $\PostBQP^{U,U^\dag}  = \PreciseBQP^{U,U^\dag}$.
\end{restatable}

\begin{lem}[Postselected thresholding]\label{lem:postselected_thresholding}
Let $\{Q_x\}$ be a uniform family of polynomial-size unitary circuits with
acceptance probabilities $a_x$, and let $c,s:\mathbb N\to[0,1]$ be
polynomial-time computable.  Suppose there is a polynomial
$r:\mathbb N\to\mathbb N$ such that $r(n)\geq1$ and
$c(n)-s(n)\geq2^{-r(n)}$.
Suppose $Q_x^\dagger$ can be run either directly or on a postselected branch
of inverse-exponential probability.  Then a polynomial-size postselected
circuit can distinguish $a_x\geq c(n)$ from $a_x\leq s(n)$.
\end{lem}

\begin{proof}
Fix an input $x$ of length $n$, and write
$Q=Q_x$, $a=a_x$, $c=c(n)$, $s=s(n)$, and $r=r(n)$.
Run $Q$, copy its output qubit, run $Q^\dagger$, and postselect the registers of
$Q$ back to their initial state.  The copied qubit is then in the
state (up to normalization)
\[
    (1-a)\ket{0}+a\ket{1}.
\]
If running $Q^\dagger$ itself uses postselection, include those outcomes in the
final postselection event.  They only multiply this state by a common factor.
Compute a dyadic rational $\eta$ with $O(r)$ bits such that
$|\eta-(c+s)/2|\leq 2^{-r}/8$, and set $z=a-\eta$.
Thus $z$ is positive in the first case, negative in the second, and
$|z|\geq 3\cdot 2^{-r}/8$.

For each $\lambda\in\{1,1/2,\ldots,2^{-(r+2)}\}$, consider the linear map
\[
    M_{\eta,\lambda}
    =\frac12
      \begin{pmatrix}
        -\eta & 1-\eta\\
        \lambda & \lambda
      \end{pmatrix}.
\]
Because $M_{\eta,\lambda}$ has norm at most one, it has a constant-dimensional
unitary dilation whose entries are polynomial-time computable from $\eta$ and
$\lambda$.  Ideally applying this dilation to an ancilla in state $\ket{0}$
and postselecting the ancilla on outcome $0$ applies
$M_{\eta,\lambda}$.  Conditioned on this outcome, the original qubit is
proportional to $z\ket{0}+\lambda\ket{1}$.
Measuring in the Hadamard basis gives bias
\[
    \frac{z\lambda}{z^2+\lambda^2}.
\]
All these biases have the sign of $z$, and one has magnitude at least $2/5$
because some $\lambda$ lies between $|z|$ and $2|z|$.  Run the construction
on a separate copy for each of the $O(r)$ values of $\lambda$, postselect on
all of their ancillas yielding $0$, and then choose one of the resulting
output qubits uniformly.  This gives
inverse-polynomial bias, which polynomially many repetitions amplify to
bounded error with, say, completeness at least $5/6$ and soundness at most
$1/6$.  For the dilation associated with a particular $\lambda$, its
postselection probability is at least $\lambda^2/4\geq 2^{-2r-6}$.
Together with the assumed inverse-exponential success probability for any
implementation of $Q^\dagger$, all postselections in the amplified ideal
circuit therefore have joint probability at least
$\mu\coloneqq 2^{-q(n)}$ for some polynomial $q$.

It remains to implement the dilations over the fixed gate set.  Apply
\autoref{lem:finite_gate_postselection} with total trace-distance error
$\delta\leq\mu/100$.  The compiled circuit still postselects with probability
at least $99\mu/100=2^{-\poly(n)}$, and its conditional acceptance
probability changes by at most $1/50$.  Hence the bounded-error inequalities
remain valid.  Since $\log(1/\delta)=\poly(n)$, this compilation has only
polynomial overhead.
\end{proof}

\begin{proof}[Proof of \autoref{thm:equal_with_inverse}]
Let a postselected circuit have postselection probability $a$ and conditional
acceptance probability $b$.  An ordinary circuit can output its answer when
postselection succeeds and a fair bit otherwise.  Its acceptance probability is
\[
    a b+\frac{1-a}{2}=\frac12+a\left(b-\frac12\right).
\]
Since $a$ is at least inverse exponential, this gives an inverse-exponential
gap around $1/2$.  Hence
$\PostBQP^{U,U^\dagger}\subseteq\PreciseBQP^{U,U^\dagger}$.

For the other direction, let $Q$ be a \PreciseBQP{} circuit.  Its adjoint is
available: reverse the known gates and replace each query to $U$ by a query to
$U^\dagger$, and conversely.  The lemma turns its inverse-exponential gap into
bounded error after postselection.  Thus
$\PreciseBQP^{U,U^\dagger}\subseteq\PostBQP^{U,U^\dagger}$.
\end{proof}

\subsection{Conjugate access}
Conjugate access lets us run the inverse of an oracle query by teleportation and
postselection.  We can then use the thresholding lemma from the previous
subsection.

\begin{restatable}{thm}{thmEqualWithConjugate}\label{thm:equal_with_conjugate_access}
    For any quantum unitary oracle $U$, $\PostBQP^{U,\bar{U}}  = \PreciseBQP^{U,\bar{U}}$.
\end{restatable}

\begin{proof}
We first show
$\PreciseBQP^{U,\bar U}\subseteq\PostBQP^{U,\bar U}$.  Let $Q$ be a
\PreciseBQP{} circuit.  To use \autoref{lem:postselected_thresholding}, we only
need to run $Q^\dagger$ on a postselected branch.

This can be achieved by a standard probabilistic unitary-transposition protocol based on \cite{quintino2019probabilistic}. Applying it to \(\bar U\) implements \((\bar U)^T=U^\dagger\) with success probability \(d^{-2}\).

Let $d$ be the dimension of $U$, and let
\[
    \ket{\Phi}=\frac{1}{\sqrt d}\sum_{j=1}^d\ket{j,j}.
\]
Preparing $\ket\Phi$ and projecting onto it use only $O(\log d)$ ordinary
gates.
To apply $U^\dagger$ to a state $\ket\psi_S$, prepare $\ket\Phi_{AB}$,
apply $\bar U$ to $A$, and postselect $S,A$ onto $\ket\Phi$.  The state left
on $B$ is
\[
 (\bra\Phi_{SA}\otimes I_B)
 \bigl(\ket\psi_S\otimes(\bar U_A\otimes I_B)\ket\Phi_{AB}\bigr)
 =\frac1d U^\dagger\ket\psi_B.
\]
This postselection succeeds with probability $d^{-2}$.  The same construction
with $U$ in place of $\bar U$ applies $U^T=(\bar U)^\dagger$.  Thus every
oracle gate in $Q^\dagger$ can be applied using the allowed queries.  If $Q$
makes $t$ queries, all the teleportations succeed with probability $d^{-2t}$,
which is inverse exponential because $\log d$ and $t$ are polynomial.  The
thresholding lemma now gives a $\PostBQP^{U,\bar U}$ circuit.

For the other containment, suppose a postselected circuit has postselection
probability $a$ and conditional acceptance probability $b$.  Output its
answer when postselection succeeds and a fair bit otherwise.  The acceptance
probability becomes
\[
    a b+\frac{1-a}{2}=\frac12+a\left(b-\frac12\right).
\]
Since $a$ is at least inverse exponential, this is a \PreciseBQP{} circuit.
Therefore $\PostBQP^{U,\bar U}=\PreciseBQP^{U,\bar U}$.
\end{proof}

\section{Equivalence with poly-dimension oracles } \label{sec:poly_dim_equiv}
We now turn to showing that these two classes are equivalent under any unitary oracle of polynomial dimension $d$. This follows from the observation that, up to a global phase, the complex conjugate representation $U^*$ of a unitary $U$ appears inside its $d-1$-fold tensor product $U^{\otimes (d -1)}$. In other words, $d-1$ queries to the unitary $U$ suffice to perfectly simulate queries to $U^*$. This is a fundamental fact about the representation theory of the unitary group, and was first viewed as an algorithm in~\cite{miyazaki2019complex}. For completeness we present here the algorithm as well as a short proof from first principles.

For each $x \in [d]$, let 
\[
    \ket{\bar{x}} 
    := 
    \frac{1}{\sqrt{(d-1)!}} 
    \sum_{
        \substack{
            \pi \in S_{d}
            \\
            \pi(d) = x
        }
    }
    \mathsf{sign}(\pi)
    \,
    \ket{\pi(1), \dots, \pi(d-1)}
    \,.
\]
That is, it is an antisymmetrization of all the values in the standard basis \emph{besides} $\ket{x}$.

Let $V$ be a partial isometry defined as \(V := \sum_{x \in [d]} \ketbra{\bar{x}}{x}\).
For every $\delta>0$, $V$ and $V^\dagger$ can be implemented to operator-norm
error at most $\delta$ using
$\poly(d,\log(1/\delta))$ gates.  For example, this follows from the efficient
approximate Schur transform~\cite{BCH05}, followed by compilation of its
uniformly specified local gates over our fixed gate set.  We use $V$ below for
the exact isometry and denote its circuit approximation by $\widetilde V$.

We claim that the complex conjugate (in the standard basis) of the unitary $U$ is given exactly by
\begin{align}
    U^* = \frac{1}{\det(U)} V^{\dagger} \; U^{\otimes (d-1)} \; V
    \label{eq:u-complex-conjugate}
\end{align}

\begin{proof}
    We left-multiply both sides of \autoref{eq:u-complex-conjugate} by $\det(U) \cdot U^T$ and show that 
    \begin{align*}
        U^T V^{\dagger} \; U^{\otimes (d-1)} \; V = \det(U) \cdot I
    \end{align*}

    First, consider the matrix elements \(\bra{\bar{y}}U^{\otimes (d-1)} \ket{\bar{x}}\) of the tensor product:

    \begin{align*}
        &
        \bra{\bar{y}}
        U^{\otimes (d-1)} 
        \ket{\bar{x}}
        \\
        &
        =
        \frac{1}{{(d-1)!}} 
        \sum_{
            \substack{
                \pi, \sigma \in S_{d}
                \\
                \pi(d) = y,
                \; 
                \sigma(d) = x
            }
        } 
        \mathsf{sign}(\pi \sigma^{-1}) 
        \bra{\pi(1), \dots, \pi(d-1)}
        \;
        U^{\otimes (d-1)} 
        \;
        \ket{\sigma(1), \dots, \sigma(d-1)}
        \allowdisplaybreaks
        \\
        &
        =
        \frac{1}{{(d-1)!}} 
        \sum_{
            \substack{
                \pi, \sigma \in S_{d}
                \\
                \pi(d) = y,
                \; 
                \sigma(d) = x
            }
        } 
        \mathsf{sign}(\pi \sigma^{-1}) 
        \prod_{
            i \in [d-1]
        }
        \bra{\pi(i)}
        U 
        \ket{\sigma(i)}
        \allowdisplaybreaks
        \\
        &
        =
        \frac{1}{{(d-1)!}} 
        \sum_{
            \substack{
                \pi, \sigma \in S_{d}
                \\
                \pi(d) = y,
                \; 
                \sigma(d) = x
            }
        } 
        \mathsf{sign}(\pi \sigma^{-1}) 
        \prod_{
            i \in [d] \setminus {x}
        }
        \bra{\pi \circ \sigma^{-1}(i)}
        U 
        \ket{i}
        \allowdisplaybreaks
        \\
        &
        =
        \sum_{
            \substack{
                \pi \in S_{d}
                \\
                \pi(x) = y,
            }
        } 
        \mathsf{sign}(\pi) 
        \prod_{
            i \in [d] \setminus {x}
        }
        \bra{\pi(i)}
        U 
        \ket{i}
    \end{align*}

    We now compute:
    \begin{align*}
        &
        \bra{z} U^T V^{\dagger} \; U^{\otimes (d-1)} \; V \ket{x}
        \allowdisplaybreaks
        \\
        &
        =
        \sum_{y \in [d]}
        \bra{z} U^T \ketbra{y}{y} V^{\dagger} \; U^{\otimes (d-1)} \; V \ket{x}
        \allowdisplaybreaks
        \\
        &
        =
        \sum_{y \in [d]}
        \bra{y} U \ketbra{z}{\bar y} U^{\otimes (d-1)}\ket{\bar x}
        \allowdisplaybreaks
        \\
        &
        =
        \sum_{y \in [d]}
        \bra{y} U \ket{z}
        \sum_{
            \substack{
                \pi \in S_{d}
                \\
                \pi(x) = y,
            }
        } 
        \mathsf{sign}(\pi) 
        \prod_{
            i \in [d] \setminus {x}
        }
        \bra{\pi(i)}
        U 
        \ket{i}
    \end{align*}

    Whenever $z = x$, then this is
    \begin{align*}
        \sum_{y \in [d]}
        \bra{y} U \ket{x}
        \sum_{
            \substack{
                \pi \in S_{d}
                \\
                \pi(x) = y,
            }
        } 
        \mathsf{sign}(\pi) 
        \prod_{
            i \in [d] \setminus {x}
        }
        \bra{\pi(i)}
        U 
        \ket{i}
        &
        =
        \sum_{
            \pi \in S_{d}
        } 
        \mathsf{sign}(\pi) 
        \prod_{
            i \in [d]
        }
        \bra{\pi(i)}
        U 
        \ket{i}
        \allowdisplaybreaks
        \\
        &
        =
        \det(U)
        \,,
    \end{align*}
    and when $z \ne x$, then this instead becomes $\det(U')$ where $U'$ is the (non-unitary) matrix that is the same matrix as $U$ but with the $x$\textsuperscript{th} column replaced with a copy of the $z$\textsuperscript{th} column, and thus having determinant 0.
    This shows that $\bra{z} U^T V^{\dagger} \; U^{\otimes (d-1)} \; V \ket{x} = \det(U) \; \delta_{zx}$ as desired.
\end{proof}

Thus \autoref{eq:u-complex-conjugate} is an exact identity of linear maps and,
up to the irrelevant global phase $\det(U)^{-1}$, it simulates the complex
conjugate of a $d$-dimensional unitary using $d-1$ queries.  Its implementation
over the fixed gate set is approximate, which we now account for.

Consider the nontrivial containment obtained by applying the conjugate-access
simulation from the previous section, and first amplify its ideal
postselected circuit to completeness at least $5/6$ and soundness at most
$1/6$.  Let $t=\poly(n)$ be the number of simulated conjugate queries, and
choose a polynomial $q$ such that $\mu\coloneqq 2^{-q(n)}$ lower-bounds its
postselection probability.  Replacing
each occurrence of $V$ and $V^\dagger$ by approximations of operator-norm
error at most $\delta$ changes each simulated conjugate query by at most
$2\delta$ and hence changes the final state by at most $2t\delta$, by a hybrid
argument.  Choose $\delta\leq\mu/(200t)$.  Then
\autoref{lem:finite_gate_postselection} shows that the compiled circuit still
has inverse-exponential postselection probability and that its conditional
acceptance probability changes by at most $1/50$.  Because
$d=\poly(n)$ and $\log(1/\delta)=\poly(n)$, implementing all copies of
$\widetilde V$ and $\widetilde V^\dagger$ still takes polynomial time.  The
previous section therefore rules out a separation under unitary oracles of
polynomial dimension in the fixed finite-gate model.

\section{Quantum state oracle separation}\label{sec:state_sep}
In this section we prove a quantum state oracle separation.
\begin{restatable}{thm}{thmStateOracleSeparation}\label{thm:q_state_separation}
    There exists a quantum state oracle $\CO_\psi$, such that $\PostBQP^{\CO_\psi}  \subsetneq \PreciseBQP^{\CO_\psi} $. 
\end{restatable}

Specifically, there is a language which is contained in \PreciseBQP~but not in \PostBQP~relative to a quantum state oracle. We achieve this by proving an exponential \PostBQP~query lower bound for a promise problem over single qubit quantum state oracles. 

We first define what it means for \PreciseBQP~and \PostBQP~to solve a promise problem over quantum state oracles.

\begin{defn}
     For each $n \in \BN$, let $\CS_{\yes}^n$ and $\CS_{\no}^n$ be finite dimensional (disjoint) sets of quantum state oracles. We say that a \PreciseBQP~query algorithm $\CA$ solves the promise problem $\{(\CS_{\yes}^n,\CS_{\no}^n)\}_n$ if there exist polynomial-time computable functions $c, s : \BN \to [0, 1]$ and a polynomially bounded function $p : \mathbb{N} \to \mathbb{N}$ such that $c(n) - s(n) \ge 2^{-p(n)}$, and
    \begin{align}
        [\Pr[\CA^{\CO_{\psi}}(1^n) ~\text{accepts}] \geq c(n)] \qquad &\forall~\CO_{\psi} \in \CS_{\yes}^n \\
       [\Pr[\CA^{\CO_{\psi}}(1^n) ~\text{accepts}] \leq s(n)] \qquad &\forall~\CO_{\psi} \in \CS_{\no}^n
    \end{align}
\end{defn}

\begin{defn}
    For each $n \in \BN$, let $\CS_{\yes}^n$ and $\CS_{\no}^n$ be finite dimensional (disjoint) sets of quantum state oracles. We say that a \PostBQP~query algorithm $\CA$ solves the promise problem $\{(\CS_{\yes}^n,\CS_{\no}^n)\}_n$ if for all $n$,
    \begin{align}
        [\Pr[\CA^{\CO_{\psi}}(1^n) ~\text{accepts}] \geq 2/3] \qquad&\forall~\CO_{\psi} \in \CS_{\yes}^n \\
       [\Pr[\CA^{\CO_{\psi}}(1^n) ~\text{accepts}] \leq 1/3]\qquad &\forall~\CO_{\psi} \in \CS_{\no}^n
    \end{align}
\end{defn}

In this section, we show that there exists a promise problem over quantum state oracles which \PreciseBQP~can solve but \PostBQP~can't. By standard diagonalization arguments, \autoref{thm:q_state_separation} follows as a corollary.

\subsection{Quantum state oracle}\label{subsec:q_state_oracle}
Let us first describe the sets $\CS_{\yes}^n$ and $\CS_{\no}^n$. Let $\ket{\psi(\theta,\phi,\tau)} \coloneqq \sqrt{1-\tau} \e^{\ri \theta} \ket{0} + \sqrt{\tau}\e^{\ri \phi} \ket{1}$ be a parameterization of a single qubit state for $\tau\in [ 0,1]$ and $\theta ,\phi \in [ 0,2 \pi)$. We define $\CS_{\tau}$ as the set of quantum state oracles corresponding to the set of quantum states
\begin{align}
     \{\ket{\psi(\theta,\phi,\tau)}: \theta ,\phi \in [ 0,2 \pi)\}.
\end{align}
We consider the task of solving the promise problem $\L\{\L(\CS_{\frac{1}{2} + \epsilon_n}, \CS_{\frac{1}{2} }\R)\R\}_n$ for $\epsilon_n \coloneqq 2^{-n}$. In the following sections, we omit the dependence of $\epsilon$ on $n$ when it is clear from context.

\subsection{Proof outline}
This and the next sections are dedicated to proving the following lemma.
\begin{lem}\label{lem:promise_problem_separation}
    There exist (disjoint) sets of quantum state oracles $\CS_{\yes}^n$ and $\CS_{\no}^n$ for each $n \in \BN$, such that:
    \begin{enumerate}
        \item There exists a \PreciseBQP~query algorithm which solves the  problem $\L\{\L(\CS_{\yes}^n,\CS_{\no}^n\R)\R\}_n$. \label{eq:promise_prob_separation_eq1}
        \item No \PostBQP~algorithm solves the problem $\L\{\L(\CS_{\yes}^n,\CS_{\no}^n\R)\R\}_n$.\label{eq:promise_prob_separation_eq2}
    \end{enumerate}
\end{lem}
Item \ref{eq:promise_prob_separation_eq1} can be proved by a simple algorithm.
\begin{claim}\label{claim:preciseBQP_algo}
    There exists a \PreciseBQP~query algorithm which solves the problem described in \autoref{subsec:q_state_oracle}.
\end{claim}
\begin{proof}
    The \PreciseBQP~algorithm that distinguishes $\CS_{\frac{1}{2} }$ and $\CS_{\frac{1}{2} + \epsilon}$ is simple. Given a single copy of the input state, $\ket{\psi(\theta,\phi,\tau)}$, the algorithm measures in the standard basis. For an input state $\ket{\psi(\theta,\phi,\tau)}$, the algorithm accepts with probability 
    \begin{align}
    \L|\bra{1}\ket{\psi(\theta,\phi,\tau)}\R|^2 = \L|\bra{1} \L(\sqrt{1-\tau}\e^{\ri \theta} \ket{0} + \sqrt{\tau}\e^{\ri \phi} \ket{1}\R)\R|^2 = \tau.
    \end{align}
    Thus for any state drawn from $\CS_{\frac{1}{2}}$, the algorithm accepts with probability $\frac{1}{2}$ and for any state drawn from $\CS_{\frac{1}{2}+ \epsilon}$, the algorithm accepts with probability $\frac{1}{2}+\epsilon$ where $\epsilon = \epsilon_n = 2^{-n}$.
\end{proof}
Item \ref{eq:promise_prob_separation_eq2} results from the following exponential query lower bound for any \PostBQP~algorithm. 
\begin{lem}\label{lem:query_lower_bound}
    Consider the promise problem described in \autoref{subsec:q_state_oracle}. Any \PostBQP~ query algorithm
    requires $\Omega(2^n)$ queries to solve the problem.
\end{lem}
We defer the proof to \autoref{subsec:query_lower_bound}.
\begin{proof}[Proof of \autoref{lem:promise_problem_separation}]
    By \autoref{lem:query_lower_bound}, for any \PostBQP~query algorithm with runtime $p(n)$, there is an $n_0\in\BN$ such that the algorithm fails to solve the problem $(\CS_{\frac{1}{2} + \epsilon_n} , \CS_{\frac{1}{2} })$ with $\epsilon_n = 2^{-n}$ for all $n>n_0$. Combined with \autoref{claim:preciseBQP_algo}, the lemma follows. 
\end{proof}

\subsection{\texorpdfstring{\PostBQP}{PostBQP} query lower bound}\label{subsec:query_lower_bound}

\begin{proof}[Proof of \autoref{lem:query_lower_bound}]

First, it is necessary to understand the close relationship between \PostBQP~query algorithms and bounded, low-degree rational functions. 

Let $\CA$ be a \PostBQP~algorithm that makes $t$ queries to the quantum state oracle for the single qubit state $\ket{\psi(\theta,\phi,\tau)}$. In the following, it is convenient to define the amplitudes of the state  $\ket{\psi(\theta,\phi,\tau)}$ by $a_{\theta,\phi,\tau} = \sqrt{1-\tau}\e^{\ri \theta}$ and  $b_{\theta,\phi,\tau} = \sqrt{\tau}\e^{\ri \phi}$ so that $\ket{\psi(\theta,\phi,\tau)} = a_{\theta,\phi,\tau} \ket{0} + b_{\theta,\phi,\tau}\ket{1}$. When the context is clear, we omit the dependence on the parameters $\theta,\phi,\tau$ so that we write $\ket{\psi(\theta,\phi,\tau)} = a \ket{0} + b\ket{1}$. Let $d$ be the dimension of the algorithm. It is well known that the amplitudes of the final state $\ket{\psi_f}$ may be represented as bounded, degree $t$ multivariate polynomials $p_i, i\in [d]$ in the amplitudes of $\ket{\psi(\theta,\phi,\tau)}$, 
\begin{align}
    \ket{\psi_f} \coloneqq \sum_{i\in[d]} p_i(a,b)\ket{i} \quad \quad\text{s.t.} ~\sum_i |p_i(a,b)|^2 = 1.
\end{align}
\PostBQP~algorithms may then post-select on a measurement $M$, 
\begin{align}
    \ket{\psi_{post}} \coloneqq \frac{M \ket{\psi_f} }{\sqrt{\bra{\psi_f}M^\dag M \ket{\psi_f}}} \quad \quad \text{w.p.}  \bra{\psi_f}M^\dag M \ket{\psi_f}.
\end{align}
Without loss of generality, let $M$ be the measurement $\ket{0}$ on the first qubit, and let $Q \subset [d]$ be the subset of post-selected basis states (containing $\ket{0}$ on the first qubit). Then
\begin{align}
    \ket{\psi_{post}} = \frac{1}{\sqrt{\bra{\psi_f}M^\dag M \ket{\psi_f}}} \sum_{i\in Q} p_i(a,b)\ket{i}.
\end{align}
After post-selection, the acceptance probability may be given by the probability of measuring $\ket{0}$ on the second qubit, again without loss of generality. We denote the subset of basis states that contribute to the acceptance probability by $\CP \subset \CQ$, and the final acceptance probability of the algorithm is
\begin{align}
    \Pr[\CA^{\CO_\psi}(1^n)~accepts]= \frac{\sum_{i \in \CP} |p_i(a,b)|^2 }{\sum_{i \in \CQ} |p_i(a,b)|^2}.
\end{align} 
Note that $ 0 \leq \sum_{i \in \CP} |p_i(a,b)|^2 \leq \sum_{i \in \CQ} |p_i(a,b)|^2 \leq 1$. We show that any rational function of this form cannot at the same time be $> 2/3 $ on all input $a,b$ corresponding to $\CS_{\frac{1}{2}+\epsilon}$ and $<1/3 $ on all input $a,b$ corresponding to $\CS_{\frac{1}{2}}$.
    
Our proof strategy is to bound the derivative of the rational function representing the acceptance probability, which will then limit its sensitivity to small changes in $\tau$.
First we show that this worst case statement reduces to an average case statement. For notational simplicity, let $P(a,b) \coloneqq \sum_{i \in \CP} |p_i(a,b)|^2$ and $Q(a,b) \coloneqq \sum_{i \in \CQ} |p_i(a,b)|^2$.
By definition, if $\CA$ solves the promise problem $(\CS_{\frac{1}{2}},\CS_{\frac{1}{2} + \epsilon})$,
\begin{enumerate} 
    \item $P(a,b) < \frac{1}{3} Q(a,b)$ for $\tau=\frac{1}{2}$ and for all $\theta, \phi$, and
    \item $P(a,b) > \frac{2}{3} Q(a,b)$ for $\tau=\frac{1}{2}+\epsilon$ and for all $\theta, \phi$. 
\end{enumerate}
We will abuse the notation $\BE_{\CS}$ to mean an expectation over the Haar distribution of the set $\CS$ such that $\theta, \phi$ are chosen independently and uniformly on $[0,2\pi)$. It follows that in expectation 
\begin{align}
    \BE_{\CS_\frac{1}{2}} P(a,b) < \frac{1}{3} \BE_{\CS_\frac{1}{2}} Q(a,b) \qquad \text{and} \qquad
    \BE_{\CS_{\frac{1}{2}+\epsilon} } P(a,b) > \frac{2}{3} \BE_{\CS_{\frac{1}{2}+\epsilon}} Q(a,b),
\end{align}
or written more suggestively, 
\begin{align}
    \frac{\BE_{\CS_{\tau}} P(a,b)}{\BE_{\CS_{\tau}} Q(a,b)} < \frac{1}{3} \text{ for } \tau = \frac{1}{2} \qquad  \text{and} \qquad 
    \frac{\BE_{\CS_{\tau} } P(a,b)}{\BE_{\CS_{\tau}} Q(a,b)} > \frac{2}{3}   \text{ for } \tau = \frac{1}{2}+\epsilon.\label{eq:ratio}
\end{align}
We will find that by symmetrizing over the phases, $\BE_{\CS_{\tau}} P(a,b)$ and $\BE_{\CS_{\tau}} Q(a,b)$  simplify to degree $t$ polynomials in $\tau$ with notably \emph{positive} coefficients. In turn, the positivity of these coefficients constrains the expressiveness of these polynomials \eqref{eq:ratio} to such an extent that their ratio is no longer sensitive to small changes in $\tau$. 
To emphasize dependence on $\tau$, we define 
\begin{align*}
    p_{sym}(\tau) \coloneqq \BE_{\CS_{\tau}} P(a,b) \qquad \text{and} \qquad
    q_{sym}(\tau)  \coloneqq \BE_{\CS_{\tau}} Q(a,b),
\end{align*} 
and their ratio 
\begin{align}
    r(\tau) \coloneqq \frac{p_{sym} (\tau)}{q_{sym}(\tau)} = \frac{\BE_{\CS_{\tau}} P(a,b)}{\BE_{\CS_{\tau}} Q(a,b)} = \frac{\sum_{i \in \CP}\BE_{\CS_{\tau}} |p_i(a,b)|^2 }{\sum_{i \in \CQ}\BE_{\CS_{\tau}} |p_i(a,b)|^2}.
\end{align}
Without loss of generality, let us analyze the contribution of a single amplitude $\BE_{\CS_{\tau}} |p_i(a,b)|^2$. Recall that $p_i(a,b)$ is a degree $t$ polynomial in the amplitudes $a$ and $b$ of $\ket{\psi(\theta, \phi, \tau)}$. Thus, it may be defined by coefficients $ c_{i,s,r} \in \BC$ as
\begin{align}
    p_i(a,b) \coloneqq \sum_{s,r=0}^t c_{i,s,r}a^s b^r.
\end{align}
In fact, $p_i$ should actually be \emph{homogeneous} polynomial of degree $t$ up to a shift by a constant, but we do not need to use this fact for the proof. 
\begin{align}
    \BE_{\CS_{\tau}} |p_i(a,b)|^2 &=  \BE_{\CS_{\tau}} \L(\sum_{s,r=0}^t \bar{c}_{i,s,r}\bar{a}^s \bar{b}^r \R)\L(\sum_{s,r=0}^t c_{i,s,r}a^s b^r\R) \label{eq:j1} \\
    &= \sum_{s,r,s',r'=0}^t \BE_{\CS_{\tau}} \bar{c}_{i,s',r'}\bar{a}^{s'} \bar{b}^{r'} c_{i,s,r}a^s b^r.
\end{align}
Recall, $a_{\theta, \phi, \tau}= \sqrt{1-\tau} \e^{\ri \theta}, b_{\theta, \phi, \tau} = \sqrt{\tau}\e^{\ri \phi}$. By symmetrizing over the phase $\theta$ and $\phi$, all of the cross terms from the square cannot contribute to the expectation, 
\begin{align}
    \BE_{\CS_{\tau}} \bar{c}_{i,s',r'}\bar{a}^{s'} \bar{b}^{r'} c_{i,s,r}a^s b^r &= \BE_{\CS_{\tau}} \bar{c}_{i,s',r'} c_{i,s,r} (\sqrt{1-\tau})^{s+s'} \e^{\ri \theta(s-s')}  (\sqrt{\tau})^{r+r'} \e^{\ri \phi (r-r')}\\
    &= \delta_{s,s'} \delta_{r,r'}|c_{i,s,r}|^2 (1-\tau)^s \tau^r.
\end{align}
Substituting this into \eqref{eq:j1}, we have 
\begin{align}
    \BE_{\CS_{\tau}} |p_i(a,b)|^2 = \sum_{s,r=0}^t |c_{i,s,r}|^2 (1-\tau)^s \tau^r.
\end{align}
Generalizing this calculation, 
\begin{align}
    p_{sym}(\tau) &= \sum_{i \in \CP}\BE_{\CS_{\tau}} |p_i(a,b)|^2 = \sum_{i \in \CP} \sum_{s,r=0}^t |c_{i,s,r}|^2 (1-\tau)^s \tau^r, \\ 
    q_{sym}(\tau)  &= \sum_{i \in \CQ} \BE_{\CS_{\tau}}|p_i(a,b)|^2 = \sum_{i \in \CQ} \sum_{s,r=0}^t |c_{i,s,r}|^2 (1-\tau)^s \tau^r, \\ 
\end{align}
and putting these together,
\begin{align}
    r(\tau) = \frac{p_{sym}(\tau)}{q_{sym}(\tau) } = \frac{\sum_{i \in \CP} \sum_{s,r=0}^t |c_{i,s,r}|^2 (1-\tau)^s \tau^r}{\sum_{i \in \CQ} \sum_{s,r=0}^t |c_{i,s,r}|^2 (1-\tau)^s \tau^r}.
\end{align}
Next, we bound the derivative $r'(\tau)$.  Since
$r(\frac12+\epsilon)>2/3$, the polynomial $p_{sym}$ is not identically zero.
All its coefficients are nonnegative and its monomials are strictly positive
for $\tau\in(0,1)$, so $p_{sym}(\tau),q_{sym}(\tau)>0$ throughout the relevant
interval.  Thus $r(\tau)>0$, and it suffices to bound the log-derivative.
Because $r(\tau)\leq1$,
\begin{align}
    |r'(\tau)|
    =r(\tau)\left|\frac{\diff}{\diff\tau}\log r(\tau)\right|
    \leq \left|\frac{\diff}{\diff\tau}\log r(\tau)\right|.
\end{align}
We have
\begin{align}
    \frac{\diff}{\diff \tau} \log r(\tau)
    =\frac{\diff}{\diff \tau}\left(\log p_{sym}(\tau)-\log q_{sym}(\tau)\right)
    =\frac{p_{sym}'(\tau)}{p_{sym}(\tau)}
     -\frac{q_{sym}'(\tau)}{q_{sym}(\tau)}.
\end{align}
Isolating a single summand of $p_{sym}(\tau)$,
\begin{align}
     \frac{\diff}{\diff \tau} \L(\tau^r (1-\tau)^s \R)  = \L(\frac{r}{\tau} - \frac{s}{1-\tau}\R) \tau^r(1-\tau)^s,
\end{align}
and thus
\begin{align}
    \frac{p_{sym}'(\tau)}{p_{sym}(\tau)} = \frac{\sum_{i \in \CP} \sum_{s,r=0}^t |c_{i,s,r}|^2 (1-\tau)^s \tau^r \L(\frac{r}{\tau} - \frac{s}{1-\tau}\R) }{\sum_{i \in \CP} \sum_{s,r=0}^t |c_{i,s,r}|^2 (1-\tau)^s \tau^r}
\end{align}
On the domain $\tau \in [0,1]$, observe that $(1-\tau)^s \tau^r \geq 0$. It follows that we may interpret $\frac{p_{sym}'(\tau)}{p_{sym}(\tau)}$ as a weighted sum of the terms $\L(\frac{r}{\tau} - \frac{s}{1-\tau}\R)$ with \emph{positive} weights $|c_{i,s,r}|^2 (1-\tau)^s \tau^r$. Since $r$ and $s$ are bounded by $t$, $-\frac{t}{1-\tau} \leq \frac{r}{\tau} - \frac{s}{1-\tau} \leq  \frac{t}{\tau}$ and 
\begin{align}
    \L|\frac{p_{sym}'(\tau)}{p_{sym}(\tau)}\R| &= \L|\frac{\sum_{i \in \CP} \sum_{s,r=0}^t |c_{i,s,r}|^2 (1-\tau)^s \tau^r \L(\frac{r}{\tau} - \frac{s}{1-\tau}\R) }{\sum_{i \in \CP} \sum_{s,r=0}^t |c_{i,s,r}|^2 (1-\tau)^s \tau^r}\R| \\
    &\leq \frac{\sum_{i \in \CP} \sum_{s,r=0}^t |c_{i,s,r}|^2 (1-\tau)^s \tau^r \L|\frac{r}{\tau} - \frac{s}{1-\tau}\R| }{\sum_{i \in \CP} \sum_{s,r=0}^t |c_{i,s,r}|^2 (1-\tau)^s \tau^r} \\
    &\leq \max\L(\frac{t}{1-\tau},\frac{t}{\tau}\R).
\end{align}
A similar argument gives the bound $\L|\frac{q_{sym}'(\tau)}{q_{sym}(\tau)}\R| \leq \max\L(\frac{t}{1-\tau},\frac{t}{\tau}\R)$.
Finally,
\begin{align}
    |r'(\tau)|
    \leq \L|\frac{\diff}{\diff \tau} \log r(\tau) \R|
    = \L|\frac{p_{sym}'(\tau)}{p_{sym}(\tau)} - \frac{q_{sym}'(\tau)}{q_{sym}(\tau)}\R|
    \leq \L|\frac{p_{sym}'(\tau)}{p_{sym}(\tau)}\R| +\L| \frac{q_{sym}'(\tau)}{q_{sym}(\tau)}\R|
    \leq  2\max\L(\frac{t}{1-\tau},\frac{t}{\tau}\R).
\end{align}
For the relevant range of $\tau\in [\frac{1}{2}, \frac{1}{2} + \epsilon]$, \begin{align}
    \sup_{\tau\in [\frac{1}{2}, \frac{1}{2} + \epsilon]} |r'(\tau)|
    \leq \sup_{\tau\in [\frac{1}{2}, \frac{1}{2} + \epsilon]}
    2\max\L(\frac{t}{1-\tau},\frac{t}{\tau}\R)
    \leq \frac{4t}{1 - 2\epsilon}.
\end{align}
By the mean value theorem, the rational function $r(\tau)$ cannot be too sensitive to small changes in $\tau$ around $\tau=\frac{1}{2}$,
\begin{align}
    \L|r\L(\frac{1}{2} + \epsilon\R) - r\L(\frac{1}{2}\R)\R|
    \leq \epsilon\sup_{\tau\in [\frac{1}{2}, \frac{1}{2} + \epsilon]}|r'(\tau)|
    \leq  \frac{4t\epsilon}{1 - 2\epsilon}.
\end{align}
Thus, the query complexity is lower bounded by $t=\Omega(2^n)$, in order for $\L|r\L(\frac{1}{2} + \epsilon\R) - r\L(\frac{1}{2}\R)\R| = \Omega(1)$
    
\end{proof}

\subsection{Proof of state oracle separation}
In this section, we construct an oracle $\CO_\psi^*$ from the sets $\CS_{\yes}$ and $\CS_{\no}$, such that $\PreciseBQP^{\CO_\psi^*} \not\subseteq \PostBQP^{\CO_\psi^*}$ via standard diagonalization techniques. 
\begin{proof}[Proof of \autoref{thm:q_state_separation}]\label{proof:proof_state_sep}
The oracle $\CO_\psi^*$ will encode the truth-tables of functions of different input lengths. By \autoref{lem:promise_problem_separation}, for every \PostBQP~query algorithm, there exists some $n_0 \in \BN$ such that for all $n>n_0$, the algorithm cannot meet both the completeness requirement (acceptance probability at least $2/3$ in the YES case) and the soundness requirement (acceptance probability at most $1/3$ in the NO case) on all $\CO_\psi \in \CS_{\yes} \cup \CS_{\no}$. 

It suffices to construct a quantum state oracle $\CO_\psi^*$, and a unary language $L^{\CO_\psi^*}$ such that any $\PostBQP^{\CO_\psi}$ algorithm fails on at least one input length. Standard arguments lift this statement so that any algorithm fails on infinitely many input lengths. Let 
\begin{align}
    L^{\CO_\psi^*} = {1^n : \CO_\psi^* |_n \in \CS_{\yes}}
\end{align}
where $\CO_\psi^* |_n$ represents the restriction to $\CO_\psi^*$ to inputs of length $n$. 
Let $\{M_i\}_{i \in \mathbb{N}}$ be an enumeration of \PostBQP~query algorithms. We build $\CO_\psi^*$ iteratively. For each $M_i$, let $n_i$ be the smallest input length such that $\CO_\psi^*$ has not been defined for strings of length $n_i$ and the run time of $M_i$ is less than the query lower bound. By \autoref{lem:promise_problem_separation}, there exists an $\CO_\psi \in \CS_{\yes} \cup \CS_{\no}$ on which $M_i$ violates the corresponding completeness or soundness requirement. Set $\CO_\psi^*$ at length $n_i$ to be $\CO_\psi$. By construction, every algorithm fails on at least one input length.

We note that $\PostBQP^{\CO_\psi}  \subseteq \PreciseBQP^{\CO_\psi} $ trivially. By definition, the \PostBQP~algorithm without post selection must have at least an inverse exponential completeness soundness gap. This proves the strict containment.

\end{proof}

\section{Superpoly dimensional unitary oracle separation}
Let $\T=S^1$, equipped with normalized Haar measure, and define
\[
  \langle \vw,\vv\rangle\coloneqq\sum_{i=1}^d\overline{w_i}v_i,
  \qquad
  \Phi(\vw,\vv)\coloneqq \operatorname{Re}\langle \vw,\vv\rangle.
\]
For $a\in(-d,d)\setminus\{d-2r:1\leq r\leq d-1\}$, let $D_a$ denote the
normalized Haar measure over $\T^d\times\T^d$ conditioned on
$\Phi(\vw,\vv)=a$. Here and throughout, conditioning on $\Phi=a$ is
understood in the sense of the regular conditional distribution, equivalently
the disintegration of product Haar measure with respect to $\Phi$. We consider
the distributional oracle problem of distinguishing between $D_a$ and
$D_{-a}$ for $a = O(n)$.

\subsection{\texorpdfstring{\PreciseBQP}{PreciseBQP} algorithm}

The following one-query Forrelation-like interference test suffices. Note that in the following we don't have an average case distinguisher, but rather a perfect distinguisher. That is, when viewed as a promise problem, there is a \PreciseBQP~algorithm that distinguishes $D_a$ from $D_{-a}$.

\begin{prop}\label{prop:unitary-promise-in-precisebqp}
For every $d$ and every $a>0$ for which $D_a$ and $D_{-a}$ are defined, the
promise problem of distinguishing $D_a$ from $D_{-a}$ admits a one-query
quantum algorithm with completeness--soundness gap $a/d$.  In particular,
taking $d=2^n$ and $a=1$ at input length $n$ gives a \PreciseBQP~algorithm.
\end{prop}

\begin{proof}
View the oracle as acting on a block qubit and an index register according to
\[
  U\ket{0,i}=w_i\ket{0,i},
  \qquad
  U\ket{1,i}=v_i\ket{1,i},
  \qquad i\in[d].
\]
Let $\ket{u_d}=d^{-1/2}\sum_{i=1}^d\ket{i}$.  The algorithm prepares
$\ket{+}\ket{u_d}$, applies $U$, applies a Hadamard gate to the block qubit,
and accepts if that qubit is measured to be $0$.

The state after the oracle query is
\[
  \frac{1}{\sqrt{2d}}\sum_{i=1}^d
  \left(w_i\ket{0,i}+v_i\ket{1,i}\right).
\]
After the Hadamard gate, it becomes
\[
  \frac{1}{2\sqrt d}\sum_{i=1}^d
  \left((w_i+v_i)\ket{0,i}+(w_i-v_i)\ket{1,i}\right).
\]
Consequently, its acceptance probability is
\begin{align*}
  p_{\mathrm{acc}}(U)
    &=\frac{1}{4d}\sum_{i=1}^d|w_i+v_i|^2 \\
    &=\frac12+\frac{1}{2d}
      \operatorname{Re}\sum_{i=1}^d\overline{w_i}v_i \\
    &=\frac12+\frac{\Phi(\vw,\vv)}{2d}.
\end{align*}
Thus
\[
  p_{\mathrm{acc}}(U)=
  \begin{cases}
    \dfrac12+\dfrac{a}{2d}, & U\in\operatorname{supp}(D_a),\\[4pt]
    \dfrac12-\dfrac{a}{2d}, & U\in\operatorname{supp}(D_{-a}).
  \end{cases}
\]
We may therefore take
\[
  c=\frac12+\frac{a}{2d},
  \qquad
  s=\frac12-\frac{a}{2d},
  \qquad
  c-s=\frac{a}{d}.
\]

For the oracle separation, at input length $n\geq1$ take $d=2^n$ and
$a=1$.  Since $d$ is even and all critical values $d-2r$ are even, $D_1$ and $D_{-1}$ are well-defined.
The oracle acts on $n+1$ qubits, and $\ket{u_d}$ is prepared exactly by
applying a Hadamard gate to each of the $n$ index qubits.  The resulting
thresholds $c=1/2+2^{-(n+1)}$ and
$s=1/2-2^{-(n+1)}$ are polynomial-time computable, and
$c-s=1/d=2^{-n}$, as required for \PreciseBQP. 
\end{proof}

\subsection{\texorpdfstring{\PostBQP}{PostBQP} query lower bound}\label{subsec:unitary_lower_bound}
The key of this argument is that low-query algorithm sees only a small part of the oracle. After the coordinates it ``sees'' are fixed, many coordinates remain. Their real parts add up to a broad distribution that is close to
a Gaussian near its center.  Conditioning this sum to equal $a$ rather than $-a$ therefore has very little effect on anything that depends on at most $t$ coordinates.  

To make this idea precise, recall that the algorithm has no access to
$U^\dagger$ or $\overline U$.  Each joint postselection-and-decision
probability is therefore a sum of squared moduli of holomorphic polynomials,
meaning polynomials that use the entries of $U$ but not their complex
conjugates.  Their degree is at most the number $t$ of oracle queries.  We
prove that the expectation of any such sum of squares under $D_a$ and
$D_{-a}$ differs by at most a multiplicative factor of
$\exp(O(|a|t/d))$.  When $d=2^n$ and $t=\operatorname{poly}(n)$, this factor
is $1+o(1)$.  A correct \PostBQP~distinguisher, however, would force the ratio
of its joint acceptance and rejection probabilities to change by a factor of
at least four.  These two conclusions are incompatible for large $n$.

The lower bound proceeds as follows. Note that the third step is unique/non-standard, and the entire argument relies heavily on the multiplicative error shown.

\begin{enumerate}
    \item Symmetrization: We first average over the phases shared by $\vw$ and $\vv$.  This
    breaks the acceptance polynomial into nonnegative terms, each of which
    depends on at most $t$ variables $z_i=\overline{w_i}v_i$.
    \item Once these variables are fixed, the condition
    $\sum_i\operatorname{Re}z_i=\pm a$ is controlled by the $\leq t$ coordinates that
    remain.  Their contribution has the same density as a sum of independent
    cosines.
    \item Comparing the marginals: Near the center of that density, replacing the condition $a$ by
    $-a$ changes the marginal distribution on any subset of $\leq t$ coordinates by only a $1+o(1)$ multiplicative factor pointwise.
    \autoref{fig:relative-error-intuition} illustrates this comparison.
    \item Once we have that the marginal distribution over any subset of $\leq t$ coordinates is pointwise $1+o(1)$ multiplicatively close, this implies that the expectation of our non-negative and bounded polynomials are $1+o(1)$ multiplicatively close as well. Now by the symmetrization, this implies the expectation of the overall polynomials are $1+o(1)$ multiplicatively close as well, which is incompatible with any $\PostBQP$ algorithm.    
\end{enumerate}

Writing $\overline{w_i}v_i=e^{i\theta_i}$ gives
\[
  \Phi(\vw,\vv)=\sum_{i=1}^d\cos\theta_i.
\]
The differential of $\Phi$ vanishes precisely when
$\theta_i\in\{0,\pi\}$ for every $i$.  Hence the critical values of $\Phi$ are
$d-2r$, for $r=0,\ldots,d$, and $a$ is an interior regular value precisely when
\begin{equation}\label{eq:regular-values}
  a\in(-d,d)\setminus\{d-2r:1\leq r\leq d-1\}.
\end{equation}
This condition ensures that conditioning on $\Phi=a$ gives a well-behaved
surface measure rather than concentrating on a critical level set.

Let
\[
  P(\vw,\vv)=\sum_j|p_j(\vw,\vv)|^2,
\]
where the sum is finite and each $p_j$ is a polynomial of total degree at most
$t$ in $w_1,\ldots,w_d,v_1,\ldots,v_d$, with no conjugated variables.  Assume
throughout that $0\leq t<d$.

Let $\theta_1,\theta_2,\ldots$ be independent and uniform on $[0,2\pi)$, and
denote by $f_m$ the even density of
\[
  S_m=\sum_{i=1}^m\cos\theta_i.
\]
The next lemma is the precise comparison we need.  Its exponent is small when
the polynomial examines far fewer than $d$ coordinates.
\begin{lem}\label{thm:comparison}
For every $A>0$, there exist an integer $m_A\geq5$ and a constant
$K_A<\infty$ with the following property.  Let $d\geq1$ and $0\leq t<d$ be
integers, let $a$ be an interior regular value of $\Phi$, and suppose that
\begin{equation}\label{eq:central-window-hypotheses}
  d-t\geq m_A,
  \qquad
  |a|+t\leq A\sqrt{d-t}.
\end{equation}
Then every sum of squares $P$ defined above satisfies
\begin{equation}\label{eq:comparison}
  e^{-\delta_{d,t,a}}\BE_{D_{-a}}P
  \leq \BE_{D_a}P
  \leq e^{\delta_{d,t,a}}\BE_{D_{-a}}P,
\end{equation}
where
\begin{equation}\label{eq:comparison-delta}
  \delta_{d,t,a}
  =\frac{4|a|t}{d-t}
   +K_A\frac{t}{(d-t)^{3/2}}.
\end{equation}
In particular, $\delta_{d,t,a}=o(1)$ when $a=O(1)$,
$t=\operatorname{poly}(n)$, and $d=2^n$.
\end{lem}

\paragraph{Symmetrization.}
We begin by using the phase symmetry of the distribution.  Averaging over the
shared phases removes cross terms and leaves pieces that each involve only a
small set of coordinates.

\begin{claim}\label{lem:coordinate-reduction}
Let $\mu_a$ be normalized Haar measure on $\T^d$ conditioned on
$\sum_{i=1}^d\operatorname{Re}z_i=a$.  There exist polynomials
$F_{j,\gamma}$, each depending on at most $t$ coordinates of $z$, such that
\begin{equation}\label{eq:decomposition}
  \BE_{D_a}P
  =\sum_{j,\gamma}\BE_{\mu_a}|F_{j,\gamma}(\vz)|^2.
\end{equation}
\end{claim}

\begin{proof}
Set $z_i=\overline{w_i}v_i$, so that $v_i=w_i z_i$ and
\[
  \operatorname{Re}\langle \vw,\vv\rangle
  =\sum_{i=1}^d\operatorname{Re}z_i.
\]
The map $(\vw,\vv)\mapsto(\vw,\vz)$ preserves normalized Haar measure.
Conditional on $\sum_i\operatorname{Re}z_i=a$, the variable $\vw$ is therefore
independent of $z$ and Haar distributed on $\T^d$.  Here Haar measure is simply
the uniform, rotation-invariant measure on the torus.

For
\[
  p_j(\vw,\vv)=\sum_{\alpha,\beta}c_{j,\alpha,\beta}\vw^\alpha \vv^\beta,
\]
substitution of $\vv=\vw z$ and collection of terms with
$\gamma=\alpha+\beta$ give
\begin{align*}
  p_j(\vw,\vw \vz)
    &=\sum_{|\gamma|\leq t}\vw^\gamma F_{j,\gamma}(\vz),\\
  F_{j,\gamma}(\vz)
    &=\sum_{\beta\leq\gamma}
      c_{j,\gamma-\beta,\beta}\vz^\beta.
\end{align*}
The polynomial $F_{j,\gamma}$ depends only on the coordinates indexed by
$\Supp(\gamma)$, and
\[
  |\Supp(\gamma)|\leq|\gamma|\leq t.
\]
Orthogonality of the characters $\vw^\gamma$ yields
\[
  \BE_{\vw}|p_j(\vw,\vw z)|^2=\sum_\gamma|F_{j,\gamma}(\vz)|^2.
\]
In other words, averaging over the common phases makes all terms with
different values of $\gamma$ cancel.  Integration with respect to $\mu_a$ and
summation over $j$ prove \eqref{eq:decomposition}.
\end{proof}

\paragraph{Conditioning on the remaining coordinates.}
Each surviving term now depends on at most $t$ coordinates.  The next claim
describes the distribution of those coordinates after conditioning on the
total sum.
\begin{claim}\label{lem:conditional-marginal}
Let $0\leq k<d$, and let $H\geq0$ be a measurable function of the first $k$
coordinates of $z$.  With
\[
  s=\sum_{i=1}^k\operatorname{Re}z_i,
\]
one has
\begin{align}
  \BE_{\mu_a}H
    &=\frac{1}{f_d(a)}
      \int_{\T^k}H(\vz)f_{d-k}(a-s)\,d\lambda^k(\vz),
      \label{eq:marginal-a}\\
  \BE_{\mu_{-a}}H
    &=\frac{1}{f_d(a)}
      \int_{\T^k}H(\vz)f_{d-k}(a+s)\,d\lambda^k(\vz),
      \label{eq:marginal-minus-a}
\end{align}
where $\lambda$ is normalized Haar measure on $\T$.
\end{claim}

\begin{proof}
After the first $k$ coordinates contribute $s$, the remaining $d-k$
coordinates must contribute $a-s$.  Their sum has density $f_{d-k}$, which
accounts for the factor $f_{d-k}(a-s)$ in \eqref{eq:marginal-a}.
Disintegration with respect to $\sum_{i=1}^d\operatorname{Re}z_i$ makes this
calculation formal.  Applying the same identity at $-a$ and using the evenness
of $f_d$ and $f_{d-k}$ gives \eqref{eq:marginal-minus-a}.
\end{proof}

\paragraph{Comparing the marginals.}
For large $m$, the center of $f_m$ behaves like a Gaussian density with
variance $m/2$.  What matters here is not only the height of the density but
how quickly it changes.  The relevant quantity is its logarithmic derivative
$f_m'/f_m$, which measures relative rather than absolute change.  The
following proposition says that this logarithmic derivative is close to
$-2x/m$ in the central window, exactly as it would be for a Gaussian with
variance $m/2$.

\begin{prop}\label{prop:direct-score}
For every $A>0$, there exist a constant $B_A<\infty$ and an integer
$m_A\geq5$ such that, for every $m\geq m_A$ and $|x|\leq A\sqrt m$,
\begin{equation}\label{eq:direct-score}
  \left|
    \frac{f_m'(x)}{f_m(x)}+\frac{2x}{m}
  \right|
  \leq \frac{B_A}{m^{3/2}}.
\end{equation}
\end{prop}

The proof is given in \autoref{app:direct-score} and follows by a direct calculation using an explicit formula for the density $f_m$ in terms of Bessel functions. The estimate says that,
within the central window, $f_m$ changes at nearly the same relative rate as a Gaussian density with variance $m/2$.

We can now compare the density at the two points that appear in the conditional marginals.  Integrating the logarithmic derivative turns the local estimate above into a multiplicative comparison.

\begin{claim}\label{lem:central-density-ratio}
For every $A>0$, let $B_A$ and $m_A$ be as in
\autoref{prop:direct-score}.  If $m\geq m_A$ and
$|a|+|s|\leq A\sqrt m$, then
\begin{equation}\label{eq:direct-density-ratio}
  \left|
    \log\frac{f_m(a-s)}{f_m(a+s)}-\frac{4as}{m}
  \right|
  \leq 2B_A\frac{|s|}{m^{3/2}}.
\end{equation}
\end{claim}

\begin{proof}
The interval with endpoints $a-s$ and $a+s$ is contained in
$[-A\sqrt m,A\sqrt m]$. \autoref{prop:direct-score} gives
\begin{align*}
  \log\frac{f_m(a-s)}{f_m(a+s)}
    &=-\int_{a-s}^{a+s}\frac{f_m'(u)}{f_m(u)}\,du \\
    &=\frac{2}{m}\int_{a-s}^{a+s}u\,du+E_m(a,s) \\
    &=\frac{4as}{m}+E_m(a,s),
\end{align*}
where
\[
  |E_m(a,s)|\leq 2B_A\frac{|s|}{m^{3/2}}.
\]
This proves \eqref{eq:direct-density-ratio}.
\end{proof}

\paragraph{Implication for acceptance polynomials.}
We now have all the pieces needed for \autoref{thm:comparison}.  A term that
depends on $k\leq t$ coordinates leaves at least $d-t$ coordinates random, so
the central density estimate applies to the remaining sum.

\begin{proof}[Proof of \autoref{thm:comparison}.]
Let $H\geq0$ depend on $k\leq t$ coordinates.  Permutation invariance permits these coordinates to be relabeled as the first $k$.  Take $m_A$ and $B_A$
from \autoref{prop:direct-score}, and set $K_A=2B_A$.  For
$0\leq k\leq t$ and $|s|\leq k$, assumptions
\eqref{eq:central-window-hypotheses} imply
\[
  d-k\geq d-t\geq m_A,
  \qquad
  |a|+|s|
  \leq A\sqrt{d-t}
  \leq A\sqrt{d-k}.
\]
\autoref{lem:central-density-ratio}, applied with $m=d-k$, therefore gives
\begin{equation}\label{eq:pointwise-density-comparison}
  e^{-\delta_{d,t,a}}
  \leq
  \frac{f_{d-k}(a-s)}{f_{d-k}(a+s)}
  \leq
  e^{\delta_{d,t,a}}.
\end{equation}
Indeed, the absolute value of the logarithm of this ratio is at most
\[
  \frac{4|a||s|}{d-k}
  +2B_A\frac{|s|}{(d-k)^{3/2}}
  \leq
  \frac{4|a|t}{d-t}
  +K_A\frac{t}{(d-t)^{3/2}}.
\]

Equations \eqref{eq:marginal-a} and \eqref{eq:marginal-minus-a}, together
with the nonnegativity of $H$, now imply
\[
  e^{-\delta_{d,t,a}}\BE_{\mu_{-a}}H
  \leq \BE_{\mu_a}H
  \leq e^{\delta_{d,t,a}}\BE_{\mu_{-a}}H.
\]
Applying this inequality to every nonnegative term $|F_{j,\gamma}|^2$ in
\autoref{lem:coordinate-reduction} and summing over $(j,\gamma)$ proves
\eqref{eq:comparison}.  The final assertion follows immediately from
\eqref{eq:comparison-delta}.
\end{proof}

It remains to compare this near-indistinguishability with what a correct
postselected computation requires.  We apply the comparison theorem
separately to the joint probabilities of accepting and rejecting after
postselection.

\begin{lem}\label{lem:unitary-promise-lower-bound}
For $d=2^n$, every polynomial-query \PostBQP~algorithm fails on some oracle in
$\operatorname{supp}(D_1)\cup\operatorname{supp}(D_{-1})$ at every
sufficiently large input length $n$.
\end{lem}

\begin{proof}
Let an algorithm make $t$ queries, and let $P(U)$ and $Q(U)$ be the joint
probabilities of postselection and acceptance, and of postselection and
rejection, respectively.  Since neither $U^\dagger$ nor $\overline U$ is
available, $P$ and $Q$ are sums of squared moduli of holomorphic polynomials of
degree at most $t$.  Bounded-error correctness gives $P\geq2Q$ on
$\operatorname{supp}(D_1)$ and $Q\geq2P$ on
$\operatorname{supp}(D_{-1})$.  Thus correctness requires the balance between
acceptance and rejection to reverse across the two distributions.

For $d=2^n$ and $t=\operatorname{poly}(n)$, the hypotheses of
\autoref{thm:comparison} hold with $a=A=1$ for all sufficiently large $n$.
Fix such an $n$ and write $\delta=\delta_{d,t,1}$.  Applying the theorem once
to $P$ and once to $Q$ gives
\[
  2\BE_{D_1}Q
  \leq \BE_{D_1}P
  \leq e^\delta\BE_{D_{-1}}P
  \leq \frac{e^\delta}{2}\BE_{D_{-1}}Q
  \leq \frac{e^{2\delta}}{2}\BE_{D_1}Q.
\]
The postselection requirement gives $\BE_{D_{-1}}Q>0$, so the comparison
theorem gives $\BE_{D_1}Q>0$.  Hence correctness requires
$e^{2\delta}\geq4$.

However,
\[
  \delta_{d,t,1}
  =\frac{4t}{d-t}
   +K_1\frac{t}{(d-t)^{3/2}}
  =o(1),
\]
so the comparison theorem says that the relevant expectations become
arbitrarily close.  This contradicts $e^{2\delta}\geq4$, the fixed separation
that correctness demands.
\end{proof}
\subsection{Proof of unitary oracle separation}
The proof is essentially the same as \autoref{proof:proof_state_sep}
\begin{thm}\label{thm:unitary-oracle-separation}
There exists a complex diagonal unitary oracle $\CU$ of superpolynomial
dimension, with neither inverse nor conjugate queries, such that
\[
  \PostBQP^{\CU}\subsetneq\PreciseBQP^{\CU}.
\]
\end{thm}

\begin{proof}[Proof of \autoref{thm:unitary-oracle-separation}]
The oracle $\CU=\{U_n\}_{n\geq1}$ will encode one bit at each input length.
With $d=2^n$, let
\[
  L^{\CU}=\{1^n:U_n\in\operatorname{supp}(D_1)\}.
\]
Each $U_n$ acts on a block qubit and a $d$-dimensional index register, so its
dimension is $2^{n+1}$.

Let $M_1,M_2,\ldots$ be an enumeration of all clocked polynomial-time
\PostBQP~oracle machines in the model without inverse or conjugate queries.
We build $\CU$ iteratively.  At stage $i$, choose a fresh sufficiently large
length $n_i$.  By \autoref{lem:unitary-promise-lower-bound}, there is a
$U_{n_i}\in\operatorname{supp}(D_1)\cup\operatorname{supp}(D_{-1})$ on which
$M_i$ fails; set the length-$n_i$ component of $\CU$ to this oracle.  At every
unused length, choose any element of $\operatorname{supp}(D_{-1})$.  Hence
every \PostBQP~machine fails on at least one input, and
$L^{\CU}\notin\PostBQP^{\CU}$.

On the other hand, \autoref{prop:unitary-promise-in-precisebqp} decides
$L^{\CU}$ with one query, using thresholds
\[
  c(n)=\frac12+2^{-(n+1)},
  \qquad
  s(n)=\frac12-2^{-(n+1)},
\]
whose gap is $2^{-n}$.  Thus $L^{\CU}\in\PreciseBQP^{\CU}$.  Finally,
$\PostBQP^{\CU}\subseteq\PreciseBQP^{\CU}$ follows by outputting an unbiased
bit when postselection fails; the inverse-exponential postselection
probability gives an inverse-exponential acceptance gap.  This proves the
strict containment.
\end{proof}

\section*{AI Disclosure}
The proofs of equivalence of $\PostBQP$ and $\PreciseBQP$ with respect to various pantheons were human-generated.
GPT 5.1 and 5.6 provided assistance for the separation of $\PostBQP$ vs. $\PreciseBQP$, helping with key steps of the proof which the authors subsequently distilled and simplified. 
The authors verified the correctness and originality of all content including references.

\section*{Acknowledgments}
We thank Lijie Chen, William Kretschmer, Tony Metger, Patrick Rall, Umesh Vazirani, John Bostanci, and Mark Zhandry for helpful discussions. This material is based upon work partially supported by the National Science Foundation under award No. 2016245 and 2440805  by the Air Force Office of Scientific Research under grant agreement FA9550-21-1-0392, FA9550-24-1-0089, and FA9550-23-1-0363, by the DOE Office of Science under grant agreement DE-SC0025934, NSF awards CCF-2530159, CCF-2144219, and CCF-2329939, the Sloan Foundation, and by the Shoucheng Zhang Graduate Fellowship.

\bibliographystyle{amsalpha}

\bibliography{refs}

\appendix

\section{Proof of the central density estimate}
\label{app:direct-score}

This appendix proves \autoref{prop:direct-score}.  Recall that
\(S_m=\sum_{i=1}^m\cos\theta_i\), where the \(\theta_i\) are independent and
uniform on \([0,2\pi)\), and that \(f_m\) denotes the density of \(S_m\).

Let $J_0$ denote the Bessel function of the first kind of order zero, defined by
\[
  J_0(\xi)
  =\frac{1}{2\pi}\int_0^{2\pi}e^{i\xi\cos\theta}\,d\theta.
\]
Thus $J_0$ is the characteristic function of $\cos\theta$.  By independence,
the characteristic function of $S_m$ is $J_0(\xi)^m$, and Fourier inversion
gives
\[
  f_m(x)=\frac{1}{2\pi}\int_{\mathbb R}
  e^{-i\xi x}J_0(\xi)^m\,d\xi.
\]

We first record the regularity needed to differentiate this representation.

\begin{claim}\label{lem:fourier-regularity}
For every integer $m\geq5$, one has
$\xi J_0(\xi)^m\in L^1(\mathbb R)$ and $f_m\in C^1(\mathbb R)$.
\end{claim}

\begin{proof}
As $\xi\to+\infty$, the asymptotic expansion
\[
  J_0(\xi)=\sqrt{\frac{2}{\pi \xi}}
  \left[\cos\!\left(\xi-\frac{\pi}{4}\right)+O(\xi^{-1})\right]
\]
and the evenness of $J_0$ imply $|J_0(\xi)|\leq C|\xi|^{-1/2}$ for
$|\xi|\geq1$.  Hence
\[
  |\xi J_0(\xi)^m|\leq C^m|\xi|^{1-m/2}.
\]
For $m\geq5$, the right-hand side is integrable at infinity; local
integrability at the origin follows from the boundedness of $J_0$.  Fourier
inversion and differentiation under the integral therefore yield
$f_m\in C^1(\mathbb R)$.
\end{proof}

\begin{proof}[Proof of \autoref{prop:direct-score}.]
Put
\[
  R_m(x)=f_m'(x)+\frac{2x}{m}f_m(x).
\]
The function $R_m$ measures the error in the Gaussian differential equation.
We bound that error through the Fourier representation.
Set $X_j=\cos\theta_j$.  The characteristic function of $X_j$ is
\[
  \BE e^{i\xi X_j}
  =\frac{1}{2\pi}\int_0^{2\pi}e^{i\xi \cos\theta}\,d\theta
  =J_0(\xi).
\]
Consequently, independence gives
\[
  \BE e^{i\xi S_m}
  =\prod_{j=1}^m\BE e^{i\xi X_j}
  =J_0(\xi)^m.
\]
Fourier inversion therefore yields
\begin{equation}\label{eq:direct-bessel-inversion}
  f_m(x)=\frac{1}{2\pi}\int_{\mathbb R}
    e^{-i\xi x}J_0(\xi)^m\,d\xi.
\end{equation}

Differentiating \eqref{eq:direct-bessel-inversion} with respect to $x$ gives
\begin{equation}\label{eq:fm-prime-fourier}
  f_m'(x)
  =-\frac{i}{2\pi}\int_{\mathbb R}
    e^{-i\xi x}\xi J_0(\xi)^m\,d\xi.
\end{equation}
On the other hand, the identity
$xe^{-i\xi x}=i\,\partial_\xi(e^{-i\xi x})$ and integration by parts give
\begin{align}
  \frac{2x}{m}f_m(x)
    &=\frac{i}{\pi m}\int_{\mathbb R}
      \partial_\xi(e^{-i\xi x})J_0(\xi)^m\,d\xi \notag\\
    &=-\frac{i}{\pi m}\int_{\mathbb R}
      e^{-i\xi x}\frac{d}{d\xi}\bigl(J_0(\xi)^m\bigr)\,d\xi \notag\\
    &=-\frac{i}{\pi}\int_{\mathbb R}
      e^{-i\xi x}J_0(\xi)^{m-1}J_0'(\xi)\,d\xi.
      \label{eq:x-fm-fourier}
\end{align}
Combining \eqref{eq:fm-prime-fourier} and \eqref{eq:x-fm-fourier} gives
\begin{equation}\label{eq:score-fourier}
  R_m(x)
  =-\frac{i}{2\pi}\int_{\mathbb R}e^{-i\xi x}J_0(\xi)^{m-1}
    \bigl(\xi J_0(\xi)+2J_0'(\xi)\bigr)\,d\xi.
\end{equation}
\autoref{lem:fourier-regularity} justifies differentiation under the integral for
$m\geq5$; the bound $J_0(\xi)=O(|\xi|^{-1/2})$ also makes the boundary term in the
integration by parts vanish.  The power series at the origin are
\begin{align*}
  \xi J_0(\xi)&=\xi-\frac{\xi^3}{4}+\frac{\xi^5}{64}+O(\xi^7),\\
  2J_0'(\xi)&=-\xi+\frac{\xi^3}{8}-\frac{\xi^5}{192}+O(\xi^7).
\end{align*}
Therefore the linear terms cancel, and
\begin{equation}\label{eq:bessel-score-cancellation}
  \xi J_0(\xi)+2J_0'(\xi)
  =-\frac{\xi^3}{8}+\frac{\xi^5}{96}+O(\xi^7)
  =-\frac{\xi^3}{8}+O(\xi^5).
\end{equation}
This cancellation is what makes the final error smaller than the main
Gaussian term.  To turn it into a uniform bound, write
\[
  B(\xi)=\xi J_0(\xi)+2J_0'(\xi).
\]
The expansion $J_0(\xi)=1-\xi^2/4+O(\xi^4)$ and
\eqref{eq:bessel-score-cancellation} imply the existence of absolute constants
$\delta\in(0,1)$, $c_0>0$, and $C_0<\infty$ such that, for $|\xi|\leq\delta$,
\[
  |J_0(\xi)|\leq e^{-c_0\xi^2},
  \qquad
  |B(\xi)|\leq C_0|\xi|^3.
\]
Consequently,
\begin{align}
  \int_{|\xi|\leq\delta}|J_0(\xi)|^{m-1}|B(\xi)|\,d\xi
  &\leq C_0\int_{\mathbb R}|\xi|^3e^{-c_0(m-1)\xi^2}\,d\xi \notag\\
  &=\frac{C_0}{c_0^2(m-1)^2}
  \leq \frac{C_1}{m^2}.
  \label{eq:score-small-frequency}
\end{align}

Choose $K$ so that, for $|\xi|\geq1$,
\[
  |J_0(\xi)|\leq K|\xi|^{-1/2},
  \qquad
  |J_0'(\xi)|\leq K|\xi|^{-1/2},
\]
and fix
\[
  R\geq\max\{1,4K^2\}.
\]
The characteristic-function representation of $J_0$ implies
$|J_0(\xi)|<1$ for $\xi\neq0$: equality would force $e^{i\xi \cos\theta}$ to be
constant almost surely.  Hence continuity gives
\[
  \rho:=\sup_{\delta\leq|\xi|\leq R}|J_0(\xi)|<1.
\]
If $M_R=\sup_{|\xi|\leq R}|B(\xi)|$, the intermediate-frequency contribution is
bounded by
\begin{equation}\label{eq:score-middle-frequency}
  \int_{\delta\leq|\xi|\leq R}|J_0(\xi)|^{m-1}|B(\xi)|\,d\xi
  \leq 2R M_R\rho^{m-1}
  \leq \frac{C_2}{m^2},
\end{equation}
where the last inequality follows from
$\sup_{m\geq5}m^2\rho^{m-1}<\infty$.

For $|\xi|\geq R$, the preceding Bessel bounds imply
\[
  |B(\xi)|\leq 3K|\xi|^{1/2}.
\]
Thus, for $m\geq5$,
\begin{align}
  \int_{|\xi|\geq R}|J_0(\xi)|^{m-1}|B(\xi)|\,d\xi
  &\leq 6K^m\int_R^\infty \xi^{-(m-2)/2}\,d\xi \notag\\
  &=\frac{12K^m}{m-4}R^{-(m-4)/2} \notag\\
  &=\frac{12R^2}{m-4}
    \left(\frac{K}{\sqrt R}\right)^m
  \leq \frac{C_3}{m^2},
  \label{eq:score-large-frequency}
\end{align}
because $K/\sqrt R\leq1/2$ and
$\sup_{m\geq5}m^2 2^{-m}/(m-4)<\infty$.
It follows from \eqref{eq:score-small-frequency}--
\eqref{eq:score-large-frequency} and \eqref{eq:score-fourier} that, uniformly
in $x$ and for every $m\geq5$,
\begin{equation}\label{eq:R-bound}
  |R_m(x)|\leq \frac{C}{m^2}.
\end{equation}

We also need a lower bound on $f_m$ before dividing by it.  The variance gives
such a bound somewhere near the center, and the estimate on $R_m$ then carries
that bound across the whole central window.
Since $\BE S_m^2=m/2$, Chebyshev's inequality gives
\[
  \int_{-\sqrt m}^{\sqrt m}f_m(x)\,dx
  =\Pr\{|S_m|\leq\sqrt m\}\geq\frac12.
\]
Hence there exists $x_0\in[-\sqrt m,\sqrt m]$ such that
$f_m(x_0)\geq1/(4\sqrt m)$.  The definition of $R_m$ gives the exact
integrating-factor identity
\begin{equation}\label{eq:integrating-factor}
  e^{x^2/m}f_m(x)
  =e^{x_0^2/m}f_m(x_0)
   +\int_{x_0}^{x}e^{u^2/m}R_m(u)\,du.
\end{equation}
For $|x|\leq A\sqrt m$, \eqref{eq:R-bound} bounds the absolute value of the
integral by
\[
  \frac{C(A+1)e^{\max\{A^2,1\}}}{m^{3/2}}.
\]
Choose $m_A\geq5$ sufficiently large that this quantity is at most
$1/(8\sqrt m)$ for every $m\geq m_A$.  Then
\eqref{eq:integrating-factor} implies
\begin{equation}\label{eq:central-height}
  f_m(x)\geq\frac{e^{-A^2}}{8\sqrt m},
  \qquad |x|\leq A\sqrt m.
\end{equation}
Dividing \eqref{eq:R-bound} by \eqref{eq:central-height} proves
\eqref{eq:direct-score}, with $B_A=8Ce^{A^2}$.
\end{proof}
\end{document}